\documentclass[journal]{IEEEtran}
                              
\ifCLASSOPTIONcompsoc 
  \usepackage[caption=false,font=normalsize,labelfont=sf,textfont=sf]{subfig}
\else
  \usepackage[caption=false,font=footnotesize]{subfig}
\fi 
          
\usepackage{cite}
\usepackage{amssymb}
\usepackage{mathrsfs}
\usepackage{amsmath}
\usepackage{graphicx}
\usepackage{array}
\usepackage{multirow}
\usepackage{color} 
\usepackage{xcolor}
\usepackage{epstopdf}
\usepackage{cite}
\usepackage{amsfonts}

\usepackage[twocolumn,letterpaper]{geometry}

\newtheorem{theorem}{\bf{Theorem}}

\newtheorem{condition}{\bf{Assumption}}

\newtheorem{definition}{\bf{Definition}}
\newtheorem{example}{\bf{Example}}

\newtheorem{lemma}{\bf{Lemma}}

\newtheorem{problem}{\bf{Problem}}
\newtheorem{proposition}{\bf{Proposition}}
\newtheorem{remark}{\bf{Remark}}

\begin{document} 
	 
\title{\Large Attack-Resilient Distributed Algorithms for Exponential Nash Equilibrium Seeking}
 
\author
{Zhi Feng and Guoqiang Hu 


\thanks{
This work was supported by Singapore Ministry of Education Academic Research Fund Tier 1 RG180/17(2017-T1-002-158).
Z. Feng and G. Hu are with the School of Electrical and Electronic Engineering, Nanyang Technological University, Singapore 639798 (E-mail: zhifeng@ntu.edu.sg; gqhu@ntu.edu.sg). 
}  
} 
 
 
\maketitle 

\begin{abstract}
This paper investigates a resilient distributed Nash equilibrium (NE) seeking problem on a directed communication network subject to malicious cyber-attacks. The considered attacks, named as Denial-of-Service (DoS) attacks, are allowed to occur aperiodically, which refers to  interruptions of communication channels carried out by intelligent adversaries. In such an insecure network environment, the existence of cyber-attacks may result in  undesirable performance degradations or even the failures of distributed algorithm to seek the NE of noncooperative games. Hence, the aforementioned setting can improve the practical relevance of the problem to be addressed and meanwhile, it poses some technical challenges to the distributed algorithm design and exponential convergence analysis. In contrast to the existing distributed NE seeking results over a prefect communication network, an attack-resilient distributed algorithm is presented such that the NE can be exactly reached with an exponential convergence rate in the presence of DoS attacks. Inspired by the previous works in \cite{Hu15Cyber,Hu15IFAC,Hu15IJNRC,Hu17ACC,Hu19TCST,Hu20Tcyber}, an explicit analysis of the attack frequency and duration is investigated to enable exponential NE seeking with resilience against attacks.  
Examples and numerical simulation results are given to show the effectiveness of the proposed design. 
\end{abstract}

\vspace*{-3pt} 

\IEEEpeerreviewmaketitle

\vspace*{-8pt}
\section{Introduction}
\vspace*{-3pt}
Recently, distributed Nash equilibrium (NE) seeking of non-cooperative games has been attracting increasingly attention due to its broad applications in multi-robot systems \cite{Basar87AT}, mobile sensor networks \cite{Johansson12TAC}, smart grids \cite{Ye17Tcyber}, and so on. In contrast to early works (e.g., \cite{Shamma05TAC,Basar12TAC,Scutari04TIT}) with a complete information setting, players in distributed NE seeking have limited local information, i.e.,  each player needs to make the decision based on the local or relative information, e.g., information from its neighbors, to optimize its own objective function. Thus, the players involved in the game are required to communicate with each other over the network to estimate other players' actions.

\textit{Literature review:} various distributed continuous- and  discrete-time algorithms are developed in \cite{Ye17TAC,Ye19TAC,Pavel19TAC,Sun18,Liang17AT,Deng19TNNLS,Wang19ICCA,Xie13Conf,Pavel16TAC,Pavel19AT,Pang20TAC} to solve distributed NE seeking problems in games, in which each player cannot observe all other players' actions, but can exchange information between neighbors over an undirected graph or weight-balanced digraph. Gradient-based NE seeking algorithms are popular techniques to find the NE of games with differentiable objective functions where each player modifies its current action based on the gradient with respect to its own action. The distributed NE seeking issue with continuous-time agent dynamics is studied in  \cite{Ye17TAC,Ye19TAC,Pavel19TAC,Sun18,Liang17AT,Deng19TNNLS,Wang19ICCA}. In particular, distributed NE seeking strategies are proposed in \cite{Ye17TAC} and \cite{Ye19TAC} by combining the leader-follower based consensus algorithms and gradient-play strategies over an undirected and connected graph. The authors in \cite{Pavel19TAC} exploit some incremental passivity properties of pseudo-gradients to illustrate that the estimates of the proposed augmented gradient dynamics converge to the NE exponentially under the graph coupling condition. To remove this condition, the  two time-scale singular perturbation analysis is further developed. Distributed NE seeking of aggregative games has been studied in \cite{Liang17AT,Deng19TNNLS,Wang19ICCA}, where each player' objective function relies on its own action and an aggregate of all players' actions. 
On the other hand, discrete-time consensus algorithms are presented in \cite{Xie13Conf} to search for the NE of congestion games, where each player has the linear cost function. In \cite{Pavel16TAC}, the  gossip-based algorithm is developed to seek the NE.
An alternating direction method of multipliers' algorithm with the constant step-size is proposed in \cite{Pavel19AT} by exploiting 
the convex and smooth properties of pseudogradients. Recently, \cite{Pang20TAC} adopts consensus-based gradient-free NE seeking algorithms for limited cost function knowledge. Distributed algorithms via diminishing and constant step-sizes are studied, where the former 
ensures an almost sure NE convergence, while the latter provides convergence to the neighborhood of the NE.     

One observation in distributed NE seeking problems is that all players require to find the NE via information exchange between neighbors through a communication network. Unfortunately, due to malicious cyber-attacks such as DoS attacks, deception attacks (false data injection, replay attack), disclosure attacks   \cite{Mo14SP,Feng17AT,Yang19TAC,Chen19SMC,Zhang19Tcyber,Hu15IFAC,Hu15Cyber,Hu15IJNRC,Hu17ACC,Hu19TCST,Hu20Tcyber}, the secure network environment is hardly guaranteed in practice. The network security plays a fundamental yet critical role in successful information transmission. The malicious attacks interrupt, incorrect, or tamper transmitting information so that efficiency of distributed NE seeking algorithms is degraded significantly, and it might even lead to the failure of seeking the NE under malicious attacks. In light of wide applications of distributed algorithms in cyber-physical systems (safety-critical), and inspired by studies of security issues in many existing works (e.g., see \cite{Mo14SP,Feng17AT,Yang19TAC,Chen19SMC,Zhang19Tcyber,Hu15IFAC,Hu15Cyber,Hu15IJNRC,Hu17ACC,Hu19TCST,Hu20Tcyber}), it is highly desirable to determine how resilient distributed NE seeking algorithms will be against malicious attacks based on the fact that distributed algorithms are easily disrupted by malicious behaviors. \textit{Thus, the main objective of this work is to address attack-resilient problems of distributed NE seeking so as to provide certain safety and resilience performances against attacks. So far, few efforts are made on resilient distributed NE seeking.} 
This paper focuses on the attack-resilient research of distributed NE seeking of a non-cooperative game. In particular, an attack-resilient distributed algorithm is proposed to solve the NE seeking problem in the  multi-player non-cooperative game over a directed communication network under malicious DoS attacks. The major contributions of this paper are as follows. Firstly, to the best of our knowledge, this paper is the first work to solve this issue. The proposed distributed NE seeking algorithm is capable of exactly seeking the NE in the presence of DoS attacks. The exponential convergence of the proposed algorithm is rigorously guaranteed, provided that the frequency and duration of attacks satisfy certain bounded conditions. Moreover, compared with existing NE seeking works requiring ideal communication in \cite{Ye17TAC,Ye19TAC,Pavel19TAC,Sun18,Liang17AT,Deng19TNNLS,Wang19ICCA,Xie13Conf,Pavel16TAC,Pavel19AT,Pang20TAC}, we develop attack-resilient distributed NE seeking algorithms to search for the NE in adversarial network environments. The algorithms employ consensus-based pesudo-gradient strategies with a hybrid system method to constrain DoS  attacks. The explicit analysis is provided according to the Lyapunov stability. In addition, different from the distributed convex optimization works in \cite{Su16ACC,Sundaram19TAC,Zhao20TAC} that consider faults on nodes and require the removal of their states to be prior known, this paper investigates the more practical DoS attack on communication network. The attacks are time-sequence based and allowed to occur aperiodically. Another contribution of this work is that unlike results in \cite{Ye17TAC,Ye19TAC,Pavel19TAC,Sun18,Liang17AT,Deng19TNNLS,Wang19ICCA,Xie13Conf,Pavel16TAC,Pavel19AT,Pang20TAC} and \cite{Su16ACC,Sundaram19TAC,Zhao20TAC} over an undirected graph or weight-balanced digraph, the strongly connected directed graph is allowed here, and in the presence of malicious attacks, this directed graph can be disrupted or totally paralyzed.  
%
%

The paper is organized as follows. Section II gave  mathematical preliminaries. In Section III, the non-cooperative game and DoS attack model are described, and the main objective is presented. The attack-resilient distributed algorithm is proposed in Section IV, where the NE seeking results with an exponential convergence analysis are provided. Examples and numerical simulation results are given in Section V, followed by the conclusion in Section VI.    
 
\vspace*{-6pt} 
\section{Preliminaries} 
\vspace*{-10pt} 
\subsection{Notation \label{Notation}} 
Denote $\mathbb{R}$, $\mathbb{R}^{n}$, and $\mathbb{R}^{n\times m}$ as the sets of the real numbers, real $n$-dimensional vectors and real $n\times m$ matrices, respectively. Denote $ \mathbb{R}_{\geq 0} $ as the set of nonnegative numbers, while $ \mathbb{N^{+}} $ represents the set of positive integers.  Let $0_{n}$ ($1_{n}$) be the $n\times 1$ vector with all zeros (ones) and $I_{n}$ be the identity matrix. Let col$(x_{1},...,x_{n})$ and diag$\{a_{1},...,a_{n}\}$ be a column vector with entries $x_{i}$ and a diagonal matrix with entries $a_{i}$, $i=1,2,\cdots ,n$, respectively. The symbols $\otimes $ and $\left\Vert \cdot \right\Vert $ represent the Kronecker product and the Euclidean norm, respectively. Given a real symmetric matrix $M$, let $ M > 0 $ ($ M \geq 0 $) denote that $ M $ is positive (or positive semi-definite), and $\lambda _{\min }(M)$, $\lambda _{\max }(M)$ are its minimum and maximum eigenvalues, respectively. 
For a function $f$, it is said to be $\mathcal{C}^{m}$ if it is $ m $th continuously differentiable. For two sets ${X}$ and ${Y}$, ${Y} \backslash  {X}$  denotes the set of elements belonging to $Y,$ but not to $X.$ 

\vspace*{-11pt}
\subsection{Convex analysis \label{ConvexAnalysis}}
A function $ f: \mathbb{R}^{n} \rightarrow  \mathbb{R} $ is convex if $ f(ax+(1-a)y) \leq af(x)+(1-a)f(y)$ for any scalar $a \in [0,1] $ and vectors $ x,y \in \mathbb{R}^{n}$. $ f $ is locally Lipschitz on $ \mathbb{R}^{n}$ if it is locally Lipschitz at $ x $ for $\forall x \in \mathbb{R}^{n} $. If $ f $ is a differentiable function, $ \triangledown f $ denotes the gradient of $ f $. 
A vector-valued function (or mapping) $ F: \mathbb{R}^{n} \rightarrow \mathbb{R}^{n}$ is said to be $ \iota_{F} $-Lipschitz continuous if, for any $ x,y \in \mathbb{R}^{n}$, $ \| F(x)-F(y)\| \leq \iota_{F} \| x-y \| $. Function $ F: \mathbb{R}^{n} \rightarrow \mathbb{R}^{n}$ is (strictly) monotone if, for any $ x,y \in \mathbb{R}^{n}$, $ (x-y)^{T} (F(x)-F(y)) (>) \geq 0 $. Further, $ F $ is a $ \varepsilon $-strongly monotone, if for any scalar $ \varepsilon>0 $, and $ x,y \in \mathbb{R}^{n}$, $ (x-y)^{T} (F(x)-F(y)) > \varepsilon \|x-y\|^{2} $.

\vspace*{-11pt}
\subsection{Graph Theory\label{Graph theory}} 
Let $\mathcal{\bar{G}}$ $=$ $\left\{ \mathcal{V},\mathcal{\bar{E}}  \right\} $ denote a directed graph, where $\mathcal{V}$ $\in $ $\left\{ 1,...,N\right\} $ is a set of nodes and  $\mathcal{\bar{E}}$ $\subseteq$ $ \mathcal{V\times V}$ is a set of edges. 
An edge $(i,j)\in \mathcal{\bar{E}}$ denotes that $ i $th agent receives the information from $ j $th agent, but not vice versa. 
$\mathcal{\bar{N}}_{i}$ $=$ $%
\left\{ j\in \mathcal{V\mid }(j,i)\in \mathcal{\bar{E}}\right\} $ is a neighborhood set of the agent $i$. 
A directed graph is strongly connected if there exists a directed path connecting every pair of nodes. The matrix $\bar{A}=\left[\bar{a}_{ij}\right]$ denotes the adjacency matrix of $\mathcal{\bar{G}}$, where $\bar{a}_{ij}>0$ if and only
if $(j,i)\in \mathcal{\bar{E}}$, else $\bar{a}_{ij}=0$. A matrix $\mathcal{\bar{L}}=[\bar{l}_{ij}] \in\mathbb{R}^{N\times N}$ is called the Laplacian matrix of $\mathcal{\bar{G}}$, where $ \bar{l}_{ii}=\sum^{N}_{j=1}\bar{a}_{ij} $ and $ \bar{l}_{ij}= -\bar{a}_{ij}$, $ i \neq j$. Denote $\mathcal{\bar{L}} \triangleq \bar{D}-\bar{A}$, where $\bar{D}=\left[ \bar{d}_{ii}\right] $ is the  diagonal matrix with $\bar{d}_{ii}=\sum\nolimits_{j=1}^{N}\bar{a}_{ij}
$. 
Similarly, let $\mathcal{G}$ denote a weighted digraph, where $ a_{ij}=\omega_{i}\bar{a}_{ij} $ for a scalar $ \omega_{i}>0 $. With the same definition of $ \mathcal{\bar{L}} $, the Laplacian matrix of this new digraph  becomes $ \mathcal{L}=\mathcal{W} \mathcal{\bar{L}} $, where $\mathcal{W} = \text{diag}\{\omega_{1},\cdots,\omega_{N} \}$.

\begin{condition} \label{AssumptionGraph}
The directed graph $\mathcal{\bar{G}}$ is strongly connected. 	
\end{condition}



\begin{lemma}\label{lemma1} \cite{BookLi14}
By Assumption \ref{AssumptionGraph}, the Laplacian matrix $ \mathcal{\bar{L}} $ of 
$\mathcal{\bar{G}}$ has the following properties. 
	
\begin{enumerate}
\item $ \mathcal{\bar{L}} $ has a simple zero eigenvalue associated with right eigenvector 
$ 1_{N} $, and nonzero eigenvalues have positive real parts; 

\item $\omega = \text{col}(\omega_{1},\cdots, \omega_{N})$ with $ \omega^{T}1_{N} =1 $ is a left eigenvector of $ \mathcal{\bar{L}} $ associated with the zero eigenvalue. 
Then, $ \omega^{T}\mathcal{\bar{L}}=0^{T}_{N} $ and $ \min_{\zeta^{T}x=0, \ x\neq 0} \frac{x^{T} \mathcal{\hat{L}} x}{x^{T}x} > \frac{\lambda_{2}(\mathcal{\hat{L}})}{N}$, where $ \mathcal{\hat{L}}=\frac{1}{2} (\mathcal{L}+\mathcal{L}^{T})$ and $ \zeta $ is any vector with positive entries. Moreover, $ \omega=1_{N} $ if and only if $\mathcal{\bar{G}}$ is strongly connected and weight-balanced. 		
\end{enumerate}	
\end{lemma}

\vspace*{-2pt}   
\section{Problem Formulation}
\vspace*{-2pt} 
\subsection{Non-cooperative Game over Networks} 
\vspace*{-2pt}  
In this paper, we consider a multi-agent network consisting of $N$ players, which form a N-player non-cooperative game defined as follows. For each agent $ i \in \mathcal{V}$, the $ i $th player aims to minimize its cost function $ J_{i}(x_{i},x_{-i}): \mathbb{R}^{n_{i}} \rightarrow \mathbb{R} $ by choosing its strategy $ x_{i} \in \mathbb{R}^{n_{i}} $, and $ x_{-i}=\text{col}(x_{1},\cdots,x_{i-1},\cdots,x_{N}) $ is the strategy profile of the whole strategy profile except for player $ i $. Let $ x= \\ (x_{i},x_{-i}) $ represent all players' action profile. Alternatively, let $ x \\ =  \text{col}(x_{1},\cdots,x_{N}) \in \mathbb{R}^{n} $, $ n=\sum_{i\in \mathcal{V}}n_{i} $.   

The concept of Nash equilibrium is given below.

\begin{definition} \label{NEDefinition} \textit{(Nash equilibrium)}
A strategy profile $ x^{*}=(x^{*}_{i}, \\ x^{*}_{-i}) \in \mathbb{R}^{n}$ is said to be an Nash equilibrium of the game if 
\vspace*{-3pt}
\begin{equation}
J_{i}(x^{*}_{i},x^{*}_{-i}) \leq J_{i}(x_{i},x^{*}_{-i}), \ \text{for} \ \forall   x_{i} \in \mathbb{R}^{n_{i}}, \ i\in \mathcal{V}. \label{NECondition} 
\end{equation}
\end{definition} 

Condition (\ref{NECondition}) means that all players simultaneously take their own best (feasible) responses at the NE $ x^{*}$, where no player can unilaterally decrease its cost by changing its strategy. 


\begin{condition} \label{AssumptionConvex}
For each player $ i $, $ J_{i}(x_{i},x_{-i}) $ is $ \mathcal{C}^{2} $, strictly convex, and radially unbounded in $ x_{i} $ for each $ x_{-i} $. 
\end{condition}

Under Assumption \ref{AssumptionConvex}, it follows from \cite{Başar95} that an NE $ x^{*}$ exists, and satisfies $  \triangledown_{i} J_{i}(x^{*}_{i},x^{*}_{-i})=0_{n_{i}}$, and $ \triangledown_{i} J_{i}(x_{i},x_{-i})=\partial J_{i}(x_{i}, \\ x_{-i})/  \partial x_{i} \in \mathbb{R}^{n_{i}}$ represents the partial gradient of player $ i $'s cost with respect to its own action $ x_{i} $. We define 
\vspace*{-3pt}
\begin{equation}
F(x) \triangleq \text{col}( \triangledown_{1} J_{1}(x_{1},x_{-1}), \cdots,  \triangledown_{N} J_{N} (x_{N} \\ , x_{-N})), 
\end{equation} 
where $ F(x) \in \mathbb{R}^{n}$ denotes the pseudo-gradient (the stacked vector of all players' partial gradient). Thus, we have $ F(x^{*})=0_{n} $.

\begin{condition} \label{AssumptionSmooth}
The pesudogradient $ F $ is $ \varepsilon $-strongly monotone and $ \iota_{F} $-Lipschitz continuous for certain  constants $ \varepsilon, \iota_{F} >0 $.  
\end{condition}

\begin{remark}
Assumptions \ref{AssumptionConvex} and \ref{AssumptionSmooth} were widely used in existing works (e.g., \cite{Pavel19TAC,Sun18,Liang17AT,Deng19TNNLS,Wang19ICCA,Xie13Conf,Pavel16TAC,Pavel19AT,Pang20TAC}) to guarantee the unique NE $ x^{*} $. 
\end{remark}

In a game with perfect information (a complete communication graph) requiring the global knowledge on all players' actions, a gradient-play  algorithm  $ \dot{x}_{i}=- \triangledown_{i} J_{i}(x_{i},x_{-i})$ can be used to seek the NE asymptotically (Lemma 2, \cite{Pavel19TAC}). However, this design is impractical if the communication network is not complete. Then, a  distributed algorithm is desirable to broadcast the local information among players. Unfortunately,  when distributed computations are required and executed over insure communication networks under cyber-attacks, the exact NE may not be found.  

\subsection{Malicious DoS Attack Model}
As studied in the preliminary results \cite{Hu15IFAC,Hu15Cyber,Hu15IJNRC,Hu17ACC,Hu19TCST,Hu20Tcyber}, the well-known DoS attack refers to a class of malicious cyber-attacks where an attacker aims to corrupt and interrupt certain or all components of communication channels between the players.   
Without loss of generality, we suppose that each adversary follows an independent attacking strategy with a varying active period. Due to the limited energy, then it terminates attacking activities and shifts to a sleep period to supply its energy for next attacks.

To be specific, suppose that there exists an $m\in\mathbb{N}$ and denote $\{a^{ij}_{m}\}_{m\in\mathbb{N}}$ as an attack sequence when a DoS attack is lunched at channel $(i,j)$. Let the duration of this attack be $\tau^{ij}_{m}\geq 0$. 
Then, the $m$-th  attack strategy is generated with  $\mathcal{A}^{(i,j)}_{m} = {a^{ij}_{m}} \cup [a^{ij}_{m},a^{ij}_{m}+\tau^{ij}_{m})$ where $a^{ij}_{m+1}>a^{ij}_{m}+\tau^{ij}_{m}$ for $m\in \mathbb{N^{+}}.$ Given $t\geq t_{0} $, the sets of time instants where communication between $ (i, j) $ is denied, are described in the following form \cite{Hu17ACC,Hu19TCST,Hu20Tcyber}:  
\vspace*{-1pt}
\begin{equation}
\Xi^{(i,j)}_{a}(t_{0},t)=(\cup^{\infty}_{m=1} \ \mathcal{A}^{(i,j)}_{m})  \cap  \lbrack t_{0},t], \ m\in\mathbb{N^{+}},
\label{DoSAttacks}
\end{equation}
which implies that the sets where communication between $ (i, j) $ is allowed, are: $\Xi^{(i,j)} _{s}(t_{0},t)=[t_{0} ,t] \setminus \Xi ^{(i,j)}_{a}(t_{0}, t)$. Then,  $|\Xi_{a}(t_{0},t)|  \\ = | {\cup}_{(i,j)\in \mathcal{E}} \Xi^{(i,j)}_{a}(t_{0},t)| $ is the total length of attacker being active over $ [t_{0},t] $, while $ |\Xi_{s}(t_{0},t)|=[t_{0},t] \setminus |\Xi_{a}(t_{0},t)| $ denotes the total length of attacker being sleeping over $ [t_{0},t] $.


\vspace*{2pt} 
\begin{definition}\label{Attack Frequency}\textit{(Attack Frequency)} For any $T_{2}>T_{1}\geq t_{0}$, let $N_{a}(T_{1},T_{2})$ denote the total number of DoS attacks over $[T_{1},T_{2})$. Then,   $F_{a}(T_{1},T_{2})=\frac{N_{a}(T_{1},T_{2})}{T_{2}-T_{1}}$ denotes the attack frequency over $[T_{1},T_{2})$ for $\forall T_{2}>T_{1}\geq t_{0}$, where there exists scalars $ N_{0}, T_{f} >0 $ so that $ N_{a}(T_{1},T_{2}) \leq N_{0}+ (T_{2}-T_{1})/T_{f}. $
\end{definition}

\vspace*{2pt}
\begin{definition}\label{Attack Duration}\textit{(Attack Duration)} For any $T_{2}>T_{1}\geq t_{0},$ denote $|\Xi_{a}(T_{1},T_{2})|$ as the total time interval under attacks over $[T_{1},T_{2}).$ The attack duration over $[T_{1},T_{2})$ is defined as: there exist scalars $T_{0}\geq 0$, $T_{a}>1$ so that $|\Xi_{a}(T_{1},T_{2})| \leq T_{0}+(T_{2}-T_{1})/T_{a}.$
\end{definition}

\vspace*{3pt}
\begin{remark}
As investigated in the pioneer works  \cite{Hu15IFAC,Hu15Cyber,Hu15IJNRC,Hu17ACC,Hu19TCST,Hu20Tcyber} for multi-agent systems under attacks, Definitions \ref{Attack Frequency} and \ref{Attack Duration} are firstly introduced in \cite{Hu15IFAC} to specify attack signals in terms of frequency and time-ratio constraints. 
The $m$-th DoS attacks occurring at $a^{ij}_{m} $ between communication channels $ (i, j) $ with $ \tau^{ij}_{m} $ are allowed to occur aperiodically and can interrupt any communication channels synchronously or
asynchronously. Fig. \ref{Duration} provides an example for more details. In Definition \ref{Attack Frequency},  $1/T_{f}$ provides an upper bound on the average DoS frequency, while $ 1/T_{a} $ in Definition \ref{Attack Duration} provides an upper bound on the average DoS duration. It requires attacks to neither occur at an infinitely fast rate or be always activated.   
\end{remark}

\begin{figure}[!t]
	\centering
	\hspace{-1em}
	\includegraphics[width=8.0cm,height=2.1cm]{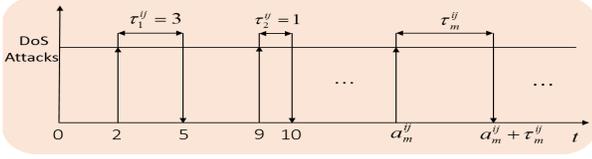}
	\caption{An illustration of DoS attacks on   communication channels $ (i,j) $ occurring at $ 2 $s, $ 9 $s, $a^{ij}_{m}$s with their durations being $ 3 $s, $ 1 $s, $ \tau^{ij}_{m} $s, respectively. Thus, we obtain $ \Xi^{(i,j)}_{a}(3,10)=[3,5)\cup [9,10)$ for just an example.} 
	\label{Duration}
\end{figure}

\vspace*{-5pt} 
\subsection{Main Objective}
\vspace*{-1pt} 
This work aims to study an attack-resilient issue of NE seeking under an insecure 
communication network as follows.

\begin{problem} \label{Problem}
(\textbf{Attack-Resilient Distributed NE Seeking}) \\
Consider a non-cooperative game consisting of $ N $ players communicating 
over an insecure communication network induced by DoS attacks. Design an attack-resilient distributed NE algorithm so that all players can exactly reach the NE $ x^{*} $ with an exponential convergence rate and a resilient feature against attacks. 
\vspace*{-3pt}
\begin{equation} \label{P} 
\hspace{1em}
\begin{split} 
& \text{minimize}   
\  J_{i}(x_{i}(t),x_{-i}(t)), \ x_{i}(t)\in \mathbb{R}^{n_{i}}, \ i=1,\cdots,N ,  \\
&\text{subject to:} \ 
\dot{x}_{i}(t)=u_{i}(t), \ t \in \Xi^{(i,j)}_{s}(t_{0},t) \cup \Xi^{(i,j)}_{a}(t_{0},t). 
\end{split}
\end{equation} 
\end{problem} 


\begin{remark}
In contrast to existing works in \cite{Ye17TAC,Ye19TAC,Pavel19TAC,Sun18,Liang17AT,Deng19TNNLS,Wang19ICCA,Xie13Conf,Pavel16TAC,Pavel19AT,Pang20TAC}, solving Problem 1 is much more challenging at least from the following aspects: \textit{(1) Player communication network:} the games involved in an insecure communication network may lead to the interruption of communication transmission caused by DoS attacks, which makes existing NE seeking algorithms  \cite{Ye17TAC,Ye19TAC,Pavel19TAC,Sun18,Liang17AT,Deng19TNNLS,Wang19ICCA,Xie13Conf,Pavel16TAC,Pavel19AT,Pang20TAC} inapplicable. \textit{(2) Assumption:} in the absence of attacks, the graph is directed rather than being an undirected graph or weight-balanced digraph in  \cite{Ye17TAC,Ye19TAC,Pavel19TAC,Sun18,Liang17AT,Deng19TNNLS,Wang19ICCA,Xie13Conf,Pavel16TAC,Pavel19AT,Pang20TAC}. 
Further, this digraph under DoS attacks can be disconnected or totally paralyzed, which brings  nontrivial convergence analysis. \textit{(3) Design requirement:} propose an attack-resilient distributed NE seeking scheme with the exponential convergence and resilience features against attacks. Due to  aforementioned challenges, existing NE seeking algorithms cannot be directly applied.
\end{remark}

\vspace*{-8pt}
\subsection{Motivating Example}
\vspace*{-2pt}
In order to illustrate the influence of DoS attacks on distributed NE seeking, the following motivation example is provided. We consider a classic example in economy, namely the Nash-Cournot game (e.g., \cite{98Agiza,Pavel19TAC}). This game includes some firms involved in the production of a homogeneous commodity, where the quantity produced by firm $ i \in \mathcal{V}$ is denoted by $ x_{i} $, and the overall cost function of each firm $ i $ is described by  
\vspace{-3pt}
\begin{equation}
J_{i}(x_{i},x_{-i})=g_{i}(x_{i})-x_{i}f(x), \ i \in \mathcal{V}, \label{MotivatedExample} \\
\end{equation} 
where $g_{i}(x_{i})=a_{i}+b_{i}(x_{i}-c_{i})+d_{i}x^{2}_{i}$ is the production cost with $ a_{i}, b_{i}, c_{i}, d_{i} $ describing the characteristics of firm $ i $, and $ f(x)=f_{0}-f_{1} \sum_{j=1}^{N}x_{j} $ is the commodity price with constants $ f_{0}, f_{1} $. As studied in \cite{Pavel19TAC}, these parameters are chosen as $ a_{i}=c_{i}=d_{i}=0 $, $ b_{i}=10+4(i-1) $, $ f_{0}=720 $, and $ f_{1}=1 $. By certain calculations, the NE is $ x^{*}=\text{col}(110,106,102,98,94,90)$. 

Next, we consider the following different cases:     
 
i) in the absence of attacks, the existing distributed NE seeking algorithm in \cite{Pavel19TAC} is performed, and Fig. \ref{MotivatingE1} shows that all players' strategies can reach consensus and converge to the NE.


\begin{figure}[!htb]
	\centering
	\begin{tabular}{c}
		\hspace*{-1.0em}		
        {\includegraphics[width=5.5cm,height=3.8cm]{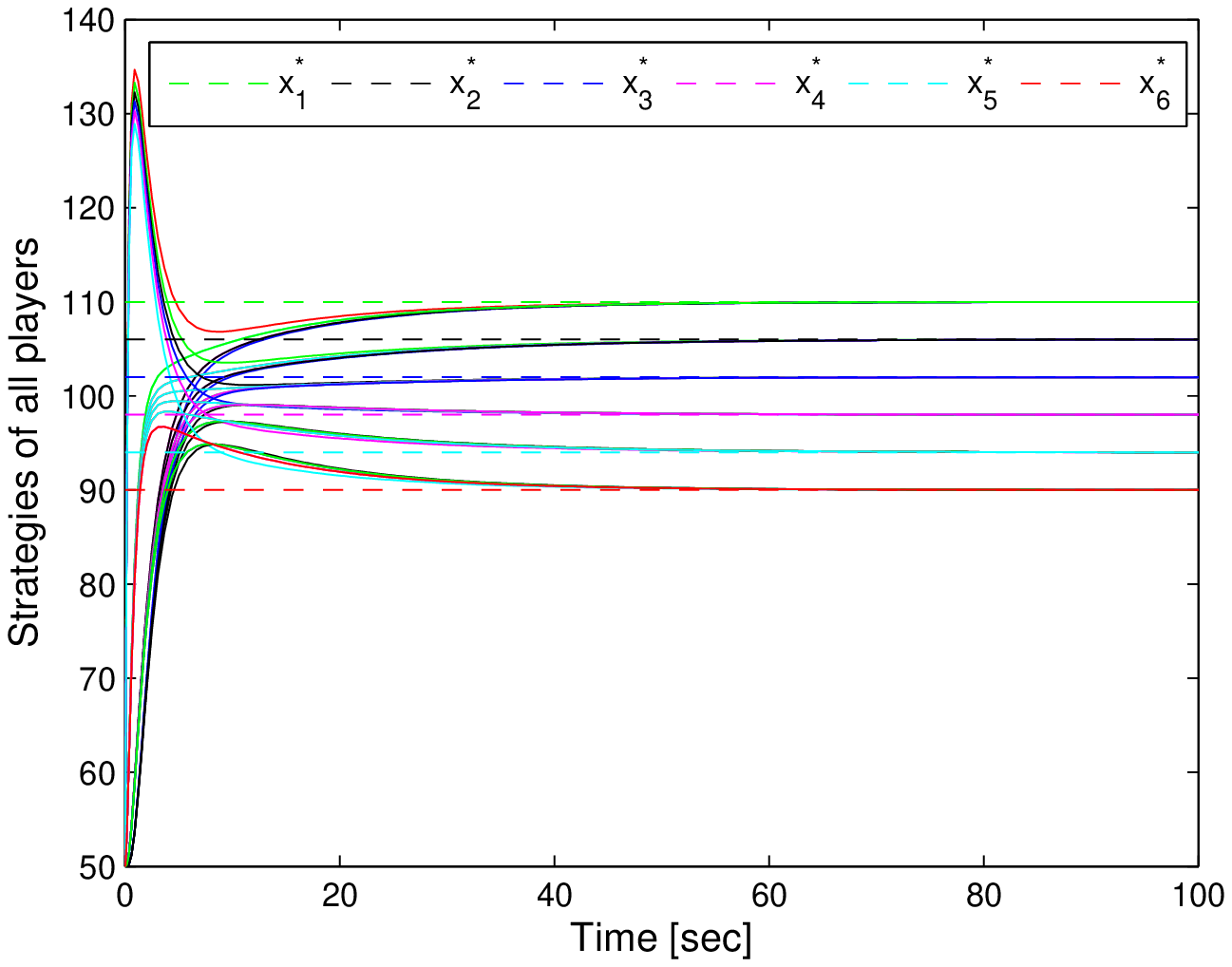} \label{m2}}
		
		\hspace*{-0.5em} 
		{\includegraphics[width=3.5cm,height=3.6cm]{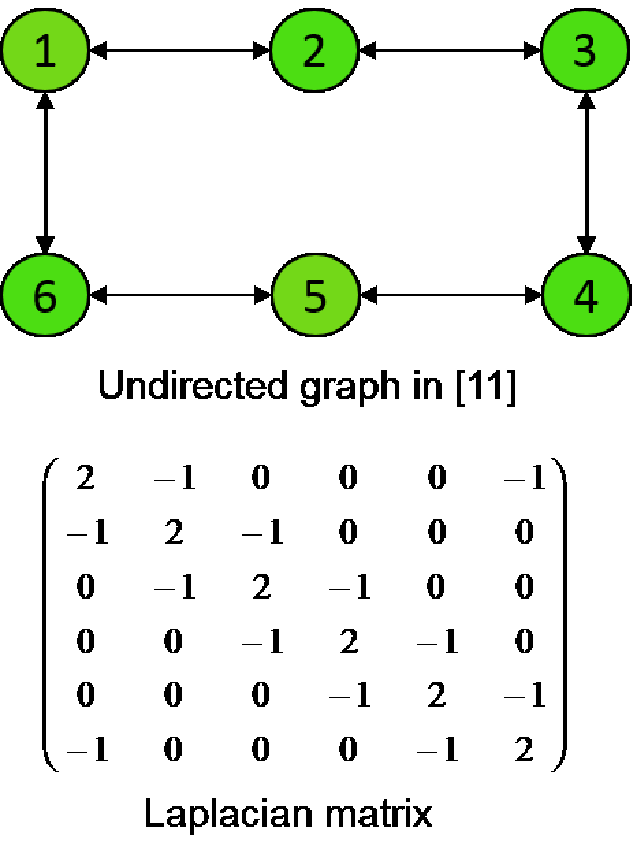} \label{m1} } 
	\end{tabular}	
	\caption{Simulated results generated by the existing algorithm \cite{Pavel19TAC} in the absence of DoS attacks. Left: plot of players' strategies and the NE $x^{*}_{i}$; Right: plot of an undirected graph  in \cite{Pavel19TAC} and its associated Laplacian matrix.}	
	\label{MotivatingE1} 
\end{figure}

ii) in the presence of attacks, this algorithm is further performed under graphs in Fig. \ref{MotivatingE2}(a) and the simulation result is shown in Fig. \ref{MotivatingE2}(b). As we see, all players' strategies neither reach consensus or converge to the NE. In contrast, the simulation result generated by the proposed attack-resilient algorithm is shown in \ref{MotivatingE2}(c), where all players' strategies reach consensus and converge to the NE.

\vspace*{-5pt}  
\begin{figure}[!htb]
	\centering
	\begin{tabular}{c}
		\hspace*{-1.2em}		
		\subfloat [Various graphs under DoS attacks on  communication channels]
		{\includegraphics[width=7.6cm,height=2.4cm]{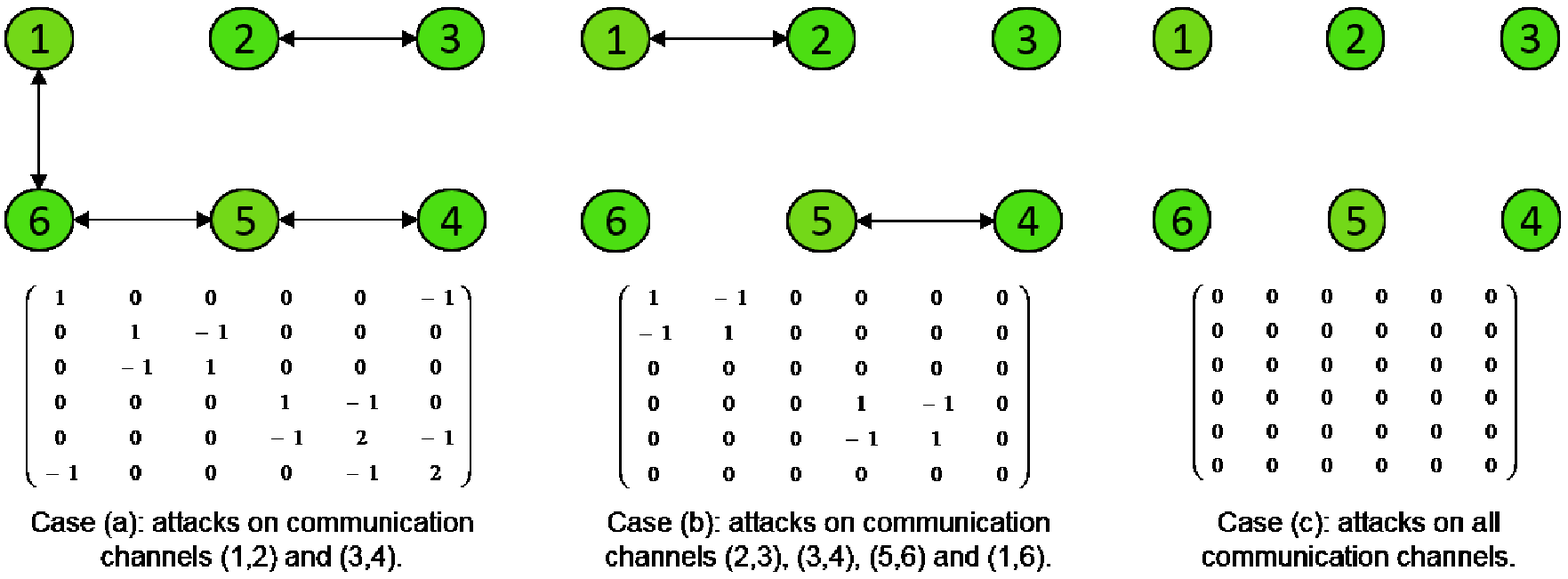} \label{me2}}
		
		\\
		
		\hspace*{-1.5em} 
		\subfloat [Algorithm in \cite{Pavel19TAC}]  
		{\includegraphics[width=5cm,height=3.2cm]{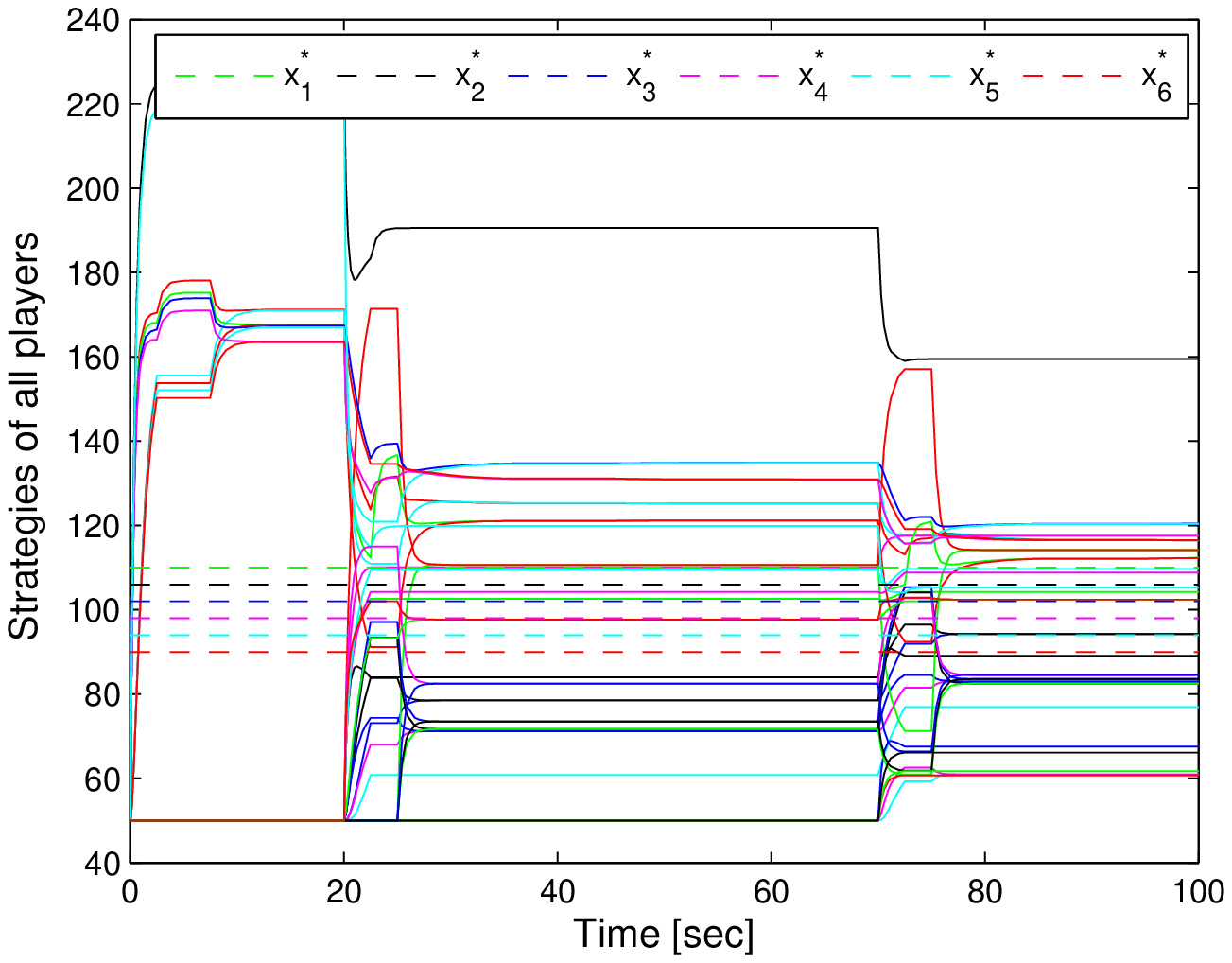} \label{me1} } 
		
		\hspace*{-1.5em} 
		\subfloat [Proposed attack-resilient algorithm]  
		{\includegraphics[width=5cm,height=3.2cm]{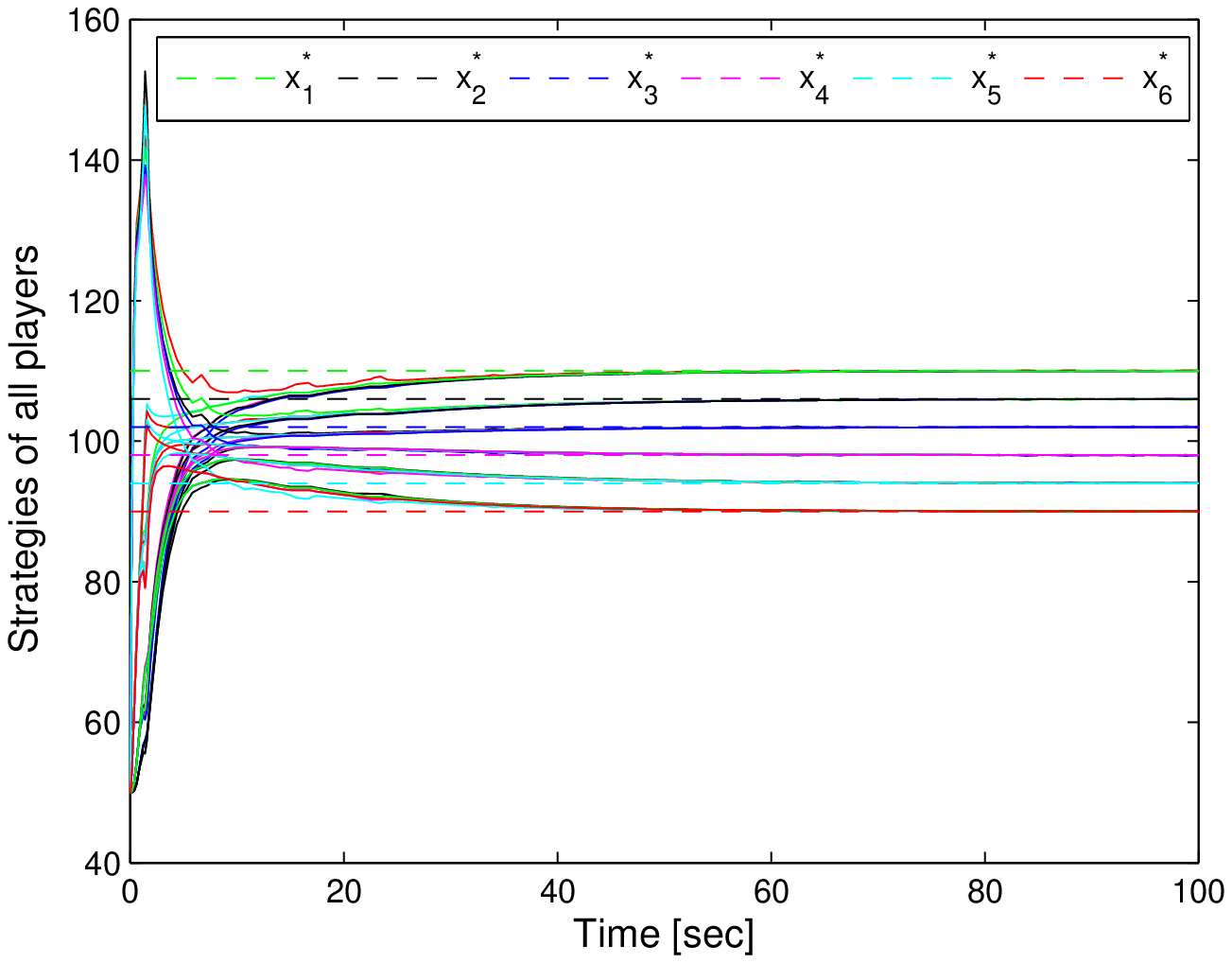} \label{me1} } 
		
	\end{tabular}	
	\caption{Simulated results generated by different algorithms under DoS attacks.} 
	\label{MotivatingE2} 
\end{figure}

\vspace*{-8pt}  
\section{Attack-Resilient Distributed NE Seeking} 
\vspace{-1pt}
In distributed NE seeking games, each player $ i $ has no access to the full information of all players' strategies. Then, each agent $ i $ shall estimate all other players' strategies. Inspired by \cite{Pavel19TAC}, let each player combine its gradient-play dynamics with an auxiliary dynamics, i.e., implement the following dynamics:  
\vspace{-3pt}
\begin{equation}
\left\{ 
\begin{array}{c}
\hspace{-0.1em}
\dot{x}^{i}_{i}=u^{i}_{i}, \ u^{i}_{i}=-\triangledown_{i} J_{i}(x^{i}_{i},\textbf{x}^{i}_{-i}) + e^{i}_{i}, \ i \in \mathcal{V},   \\ 
\hspace{-3.5em}
\dot{x}^{i}_{j} =u^{i}_{j},  \ u^{i}_{j}=e^{i}_{j}, \ \forall j \in \mathcal{V}, \ j\neq i,
\end{array}
\right. 
\label{ConsensusDesign}
\end{equation}
where player $ i $ maintains an estimate vector ${\textbf{x}}^{i}=\text{col}(x^{i}_{1}, \cdots, x^{i}_{i}, \\ \cdots, x^{i}_{N}) $ in which $ x^{i}_{j} $ is player $ i $'s estimate of player $ j $'s action, $ x^{i}_{i}=x_{i} $ is the player $ i $'s actual action, $ \textbf{x}^{i}_{-i} $ is the player $ i $' estimate vector without its own action, $ u^{i}_{i} =u_{i}$ is the player $ i $'s actual input, $ u^{i}_{j} $ is the other players' input, and $ e^{i}_{i}, e^{i}_{j} $ are to be developed. 
In (\ref{ConsensusDesign}), each player $ i $ updates $ x^{i}_{i} $ to reduce its own cost function and updates $ x^{i}_{j} $ to reach  consensus with the other players. In addition, each player $ i $ relies on its local estimated action $ \textbf{x}^{i}_{-i} $. 

For each player $ i $, (\ref{ConsensusDesign}) can be rewritten in a compact form   
\vspace{-3pt} 
\begin{equation} 	
\dot{\textbf{x}}^{i}=-\mathcal{R}^{T}_{i} \triangledown_{i} J_{i}(\textbf{x}^{i})  + \textbf{e}^{i}, \  i \in \mathcal{V},   
\label{NECompactDynamics} 
\end{equation} 
where $ \textbf{e}^{i}= \text{col}(e^{i}_{1},\cdots, e^{i}_{i}, \cdots,e^{i}_{N}) \in \mathbb{R}^{n}$ is a relative estimated error to be designed, and $ \mathcal{R}_{i} \in \mathbb{R}^{n_{i} \times n}$ used to align the gradient to the action component, is a  matrix given by 
\vspace{-5pt} 
\begin{equation} 	
\mathcal{R}_{i}=\left [0_{n_{i}\times n_{1}} \cdots 0_{n_{i}\times n_{i-1}} \ I_{n_{i}\times n_{i}} \ 0_{n_{i}\times n_{i+1}}  \cdots 0_{ n_{i}\times n_{N} } \right].  
\label{NERmatrix} 
\end{equation}

\vspace{-12pt} 
\subsection{Attack-Resilient Distributed NE Seeking Algorithm Design}
\vspace{-1pt}  
From the DoS attack model in the subsection III-B, we consider sequence-based attacks where the $ m $-th attack is lunched over a communication channel $ (i,j) $. Without loss of generality,  suppose that there exists an infinite sequence  
$k=0,1,2,\cdots$ for intervals $[t_{2k},t_{2(k+1)})$ such that in the absence of DoS attacks, each player $ i $ updates its input $ \textbf{e}^{i} $ based on an original communication network during $[t_{2k}, t_{2k+1}) $, while for each communication channel $ (i,j) $ under DoS attacks during $[t_{2k+1},t_{2k+2})$, the attackers interrupt its information transmission to make  original connected communication network disrupted or totally paralyzed. 
   
To analyze the influence of DoS attacks, we propose an \textit{attack-resilient distributed NE seeking algorithm} as  
\vspace{-2pt} 
\begin{equation} 	
\dot{\textbf{x}}^{i}=-\mathcal{R}^{T}_{i} \triangledown_{i} J_{i}(\textbf{x}^{i})  + \textbf{e}^{i},  
\ t \in [t_{2k},t_{2k+2}), \ k=0,1,\cdots,
\label{NEController1} 
\end{equation} 
where $ \textbf{e}^{i} $ denotes the relative estimated errors under attacks    
\vspace{-2pt}
\begin{equation}
\textbf{e}^{i}  = \left\{
\begin{array}{c} 
\hspace{-0.5em}
 -\kappa \sum^{N}_{j=1} a_{ij} (\textbf{x}^{i}-\textbf{x}^{j}),    \ t \in [t_{2k},t_{2k+1}), \\
\sum^{N}_{j=1}a^{\Psi(t)}_{ij} (\textbf{x}^{j}-\textbf{x}^{i}),   \ t \in  [t_{2k+1},t_{2k+2}), 
\end{array}%
\right. \label{NE2}
\end{equation}
where $\kappa$ is a positive constant gain to be specified later, and the 
$ a^{\Psi(t)}_{ij}$, $t \in  [t_{2k+1},t_{2k+2}) $,  dependent on an attack flag $ \psi(i,j,t) $, is defined as $ a^{\Psi(t)}_{ij}=0$ if $ \psi(i,j,t)=1 $ or $ -1 $; otherwise $ a^{\Psi(t)}_{ij}=1$ if $ \psi(i,j,t)=0$. The expression of $ \psi(i,j,t) $ is presented below. Set the initial value $ \psi(i,j,t)=0 $ for $ \forall (i,j)\in \bar{\mathcal{E}} $, and then each player $ i $ updates the attack flag $ \psi(j,i,t) $ as follows 

\begin{enumerate}
\item if player $ i $ can receive information from player $ j $ at $ t $,  then $ \psi(i,j,t) =0$ and it sends its information to player $ j $; 
	
\item if player $ i $ cannot receive information from player $ j $ at $ t $,  then $ \psi(i,j,t) =1$ and it sends $ \psi(i,j,t) =1$ to player $ j $; 
	
\item if player $ i $ receives $ \psi(i,j,t) =1$ which means it knows the attacking of channel $ (i,j) $, then denote $ \psi(i,j,t) =-1$ and send its information to player $ j $.  
\end{enumerate}     

\begin{remark}
To facilitate understandings of (\ref{NE2}), Fig. \ref{Time_sequence_relationship} shows the schematic of time sequences with and without DoS attacks. Intuitively, not all communication networks are secured anytime in practice, while it is reasonable to secure some original network. In the presence of DoS attacks, $ \textbf{e}^{i} $ in (\ref{NE2}) relies on $ \psi(i,j,t) $. Each player updates this attack flag once an attack signal over communication channels is detected by certain devices or mechanisms. The attack detection design is beyond the scope of this work. 
\end{remark}

\begin{remark}
In the absence of DoS attacks, $ \textbf{e}^{i} $ in (\ref{NE2}) becomes $\textbf{e}^{i} = -\kappa \sum^{N}_{j=1} a_{ij} (\textbf{x}^{i}-\textbf{x}^{j})$, which is a modified version of the design in \cite{Pavel19TAC} that requires a restrictive graph coupling condition under an undirected graph. In contrast, an adjustable proportional control gain $ \kappa $ is introduced to enable a natural trade-off between the  control effort and graph connectivity under Assumption \ref{AssumptionGraph}. 
\end{remark}

\begin{figure}[!h]
	\centering
	\includegraphics[width=8.0cm,height=1.6cm]{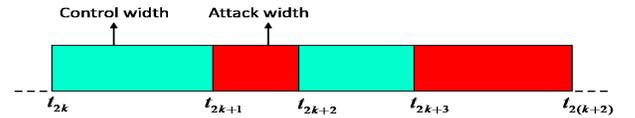}
	\caption{Schematic of time sequences with and without DoS attacks.}
	\label{Time_sequence_relationship}
\end{figure}

Next, denote the following stacked vectors and matrices 
\vspace{-2pt}
\begin{align}
\textbf{x} & =\text{col}(\textbf{x}^{1},\cdots,\textbf{x}^{N}), \  \mathcal{R}=\text{diag} \{ \mathcal{R}_{1}, \cdots,\mathcal{R}_{N} \},  \label{StackedVariable} \\
\textbf{e}&=\text{col}(\textbf{e}^{1},\cdots,\textbf{e}^{N}), \ \textbf{F(\textbf{x})} =\text{col}(\triangledown_{1} J_{1}(\textbf{x}^{1}),\cdots, \triangledown_{N} J_{N}(\textbf{x}^{N})). \notag 
\end{align}
   
Combining (\ref{NEController1})-(\ref{StackedVariable}) gives rise to the closed-loop system under attacks in the sense of a compact form  
\vspace{-3pt}
\begin{equation} \label{NEClosedCompact1}
\dot{\textbf{x}} = 
\left\{ 
\begin{array}{c} 
\hspace{-2em}
-\mathcal{R}^{T} \textbf{ F}(\textbf{x}) - \kappa (\mathcal{L}\otimes I_{n}) \textbf{x} ,  \  t \in \ [t_{2k},t_{2k+1}),  \\ 
\hspace{-0.4em}
-\mathcal{R}^{T} \textbf{ F}(\textbf{x})- (\mathcal{L}^{\Psi(t)} \otimes I_{n}) \textbf{x} ,   \   t \in \ [t_{2k+1},t_{2k+2}) ,
\end{array}
\right.	
\end{equation}
where $ \mathcal{L} $ is the Laplacian matrix of the original strongly connected digraph, while $ \mathcal{L}^{\Psi(t)}  $ is the Laplacian matrix of various potential graphs under attacks, and its zero eigenvalues may not be simple as those graphs can be unconnected under attacks.

\subsection{Stability Analysis with An Exponential Convergence Rate}
\vspace{-1pt} 
Before presenting the main result, we show that in the absence of attacks, the equilibrium of the system occurs when all players can reach consensus at the NE. 
 
\vspace{1pt}  
\begin{proposition} \label{proposition1}
Consider the game over the directed communication graph $ \bar{\mathcal{G}} $. Then, under Assumptions \ref{AssumptionGraph}-\ref{AssumptionSmooth}, $ \tilde{\textbf{x}}=1_{N}\otimes x^{*} $ is the NE of the game on networks in the absence of DoS attacks if $ \textbf{F}(\tilde{\textbf{x}}) =0_{n}  $ (or $ \triangledown_{i} J_{i}(\tilde{\textbf{x}}^{i}) =0_{n_{i}} $). At the NE,  estimated vectors of all players reach consensus and equal to the NE $ x^{*} $.
Thus, players' action components coincide with optimal actions ($ \tilde{\textbf{x}}^{i}_{i}= x^{*}_{i}$). 
\end{proposition}

\begin{proof}
In the absence of DoS attacks, let $ \tilde{\textbf{x}} $ be an equilibrium of the system. Then, we have $ 0_{Nn}=-\mathcal{R}^{T} \textbf{ F}(\tilde{\textbf{x}}) - \kappa (\mathcal{L}\otimes I_{n}) \tilde{\textbf{x}} , $
which implies that multiplying both sides by $ 1^{T}_{N} \otimes I_{n}$ yields  
\vspace{-2pt}
\begin{equation}
0_{n}=- (1^{T}_{N} \otimes I_{n}) \mathcal{R}^{T} \textbf{ F}(\tilde{\textbf{x}}) - \kappa (1^{T}_{N} \otimes I_{n}) (\mathcal{L}\otimes I_{n}) \tilde{\textbf{x}}.   \label{NE4}
\end{equation}  

Since $ 1^{T}_{N} \mathcal{L} =0^{T}_{Nn}$ under Assumption \ref{AssumptionGraph}, we obtain $ 0_{n}=(1^{T}_{N} \otimes I_{n}) \mathcal{R}^{T}  \textbf{F}(\tilde{\textbf{x}})  $. Then, it follows from the notations of $ \mathcal{R} $ and $ \textbf{F} $ in (\ref{StackedVariable}) that $ \textbf{F}(\tilde{\textbf{x}}) =0_{n}  $. Then, submitting it into (\ref{NE4}) gives rise to $ (\mathcal{L}\otimes I_{n}) \tilde{\textbf{x}}=0_{Nn} $. Hence, there exists certain $ \theta \in \mathbb{R}^{n} $ such that $ \tilde{\textbf{x}}=1_{N}\otimes \theta  $ under Assumption \ref{AssumptionGraph}. Then, it has $ \textbf{F}(1_{N}\otimes \theta) =0_{n}  $ for each player $ i \in \mathcal{V} $. Thus, $ \triangledown_{i} J_{i}(\theta_{i},\theta_{-i}) =0_{n_{i}}$. That is, $ \theta  $ is a unique NE of the game and $ \theta=x^{*} $. Thus, $ \tilde{\textbf{x}}=1_{N}\otimes x^{*} $ and for $i, j \in \mathcal{V}$,  we have $ \tilde{\textbf{x}}^{i}=\tilde{\textbf{x}}^{j}= x^{*} $ (NE of the game).
\end{proof}	


Notice that in the presence of DoS attacks, the system becomes $ \dot{\textbf{x}}= -\mathcal{R}^{T} \textbf{ F}(\textbf{x})- (\mathcal{L}^{\Psi(t)} \otimes I_{n}) \textbf{x}$. Due to the fact that the existence and uniqueness of the NE $ x^{*} $ are guaranteed under Assumptions \ref{AssumptionConvex} and \ref{AssumptionSmooth}, then, following a similar analysis above can give rise to $ 0_{n}= (1^{T}_{N} \otimes I_{n}) \mathcal{R}^{T} \textbf{ F}(\tilde{\textbf{x}})  -(1^{T}_{N} \otimes I_{n}) (\mathcal{L}^{\Psi(t)} \otimes I_{n}) \tilde{\textbf{x}}$. 
However, the presence of DoS attacks will make $ 1^{T}_{N} \mathcal{L}^{\Psi(t)} =0^{T}_{Nn}$ not hold and then, consensus estimates on the NE cannot be reached under attacks because there may not have correct information exchange among all players. Under such a situation, $ x^{*} $ may not be the NE. Next, the main task is an explicit analysis of the frequency and duration of  attacks to guarantee $ \tilde{\textbf{x}}^{i}=\tilde{\textbf{x}}^{j}= x^{*} $. 

Next, we present the main result on the resilient distributed NE seeking on networks under DoS attacks.   
 
\vspace{2pt} 
\begin{theorem} \label{Theorem2}
Under Assumptions \ref{AssumptionGraph}-\ref{AssumptionSmooth}, Problem \ref{Problem} can be solvable for any $ \textbf{x}_{i}(0)$  under the proposed resilient distributed optimization algorithm in (\ref{NEController1})-(\ref{NE2}) provided that for $ \kappa > \frac{1}{\lambda_{2}(\hat{\mathcal{L}})} (\frac{\iota^{2}}{\varepsilon} +\iota) $ and positive scalars $\lambda_{a}, \lambda_{b} , u $ to be determined later, the following two attack-related conditions are satisfied: 

(1). There exists constants $\eta^{\ast }\in (0,\lambda_{a})$ and $ \mu>1 $ so that $T_{f}$ in the \textit{attack frequency} Definition \ref{Attack Frequency} satisfies the condition:
\vspace*{-2pt}
\begin{equation}
T_{f} > T^{*}_{f}= 2 \ln(\mu)/\eta^{\ast},
\label{Condition1}
\end{equation}

(2). There exist constants $ \lambda_{a}$, $\lambda_{b}>0 $ such that $T_{a}$ in the \textit{attack duration} Definition \ref{Attack Duration} that satisfies the condition: 
\vspace*{-2pt}
\begin{equation}
\hspace{2em}
T_{a}>T^{*} _{a}=(\lambda_{a}+\lambda_{b})/(\lambda_{a}-\eta^{\ast }). 
\label{Condition2}
\end{equation}%

Moreover, the estimated states can converge to the NE with an exponential convergence rate, i.e.,
\begin{equation}
|| \textbf{x} (t) - \tilde{\textbf{x}}  ||^{2} \leq \varsigma e^{-\eta (t-t_{0})} ||\textbf{x} (t_{0}) - \tilde{\textbf{x}} ||^{2},  \ \forall t_{0} \geq 0,
\label{ExponentialConvergence}
\end{equation}%
where $ \varsigma $ is a positive scalar and 
 $ \eta=\lambda_{a}-(\lambda_{a}+\lambda_{b})/T_{a} -\eta^{*}>0$. 
\end{theorem}

\vspace{4pt}
\begin{proof}
The proof includes four steps: \\
\textbf{Step i):} when communication networks do not suffer from DoS attacks during $[t_{2k}, t_{2k+1}) $, we first show that the NE seeking can be achieved exponentially under a strongly connected digraph. 

Now, we first make a coordinate transformation as 
\vspace{-3pt}
\begin{align}
\overrightarrow{\textbf{x}} &= (1_{N} \otimes  \mathcal{S} )\textbf{x} \in \mathbb{R}^{Nn} , \ \mathcal{S} =  \frac{1}{N} (1^{T}_{N} \otimes I_{n}), \label{Xright} \\
\overleftarrow{\textbf{x}} &= (\mathcal{T} \otimes I_{n}) \textbf{x} \in \mathbb{R}^{Nn}, \ \mathcal{T}=I_{N}-\frac{1}{N} (1_{N }1^{T}_{N} ). \label{Xleft}
\end{align}
 
\vspace*{-3pt}
Then, it follows from (\ref{Xright}) that the average estimate of $ \textbf{x}^{i} $ can be described by $ \bar{\textbf{x}}=\frac{1}{N}\sum^{N}_{i=1}\textbf{x}^{i}= \frac{1}{N} (1^{T}_{N} \otimes I_{n})\textbf{x}=\mathcal{S}\textbf{x} $. Further, for any $ \textbf{x} \in \mathbb{R}^{Nn}$, it can be decomposed as  $\textbf{x}=\overrightarrow{\textbf{x}}+ \overleftarrow{\textbf{x}} $ with $ (\overrightarrow{\textbf{x}})^{T}\overleftarrow{\textbf{x}}=0 $ and $ (\mathcal{L} \otimes I_{n}) \overrightarrow{\textbf{x}}= 0_{Nn} $ under Assumption \ref{AssumptionGraph}. 

\vspace{3pt}
For stability analysis, we select the Lyapunov function as
\vspace*{-1pt}
\begin{align}
V_{a}(\textbf{x})&= \frac{\alpha}{2} (\textbf{x}-\tilde{\textbf{x}})^{T}(\textbf{x}-\tilde{\textbf{x}})= \frac{\alpha}{2} (\overrightarrow{\textbf{x}}+\overleftarrow{\textbf{x}} -\tilde{\textbf{x}})^{T}(\overrightarrow{\textbf{x}}+\overleftarrow{\textbf{x}} -\tilde{\textbf{x}}) \notag \\
&=\frac{1}{2} 
\left[
\begin{array}{l}  
  \overrightarrow{\textbf{x}}-\tilde{\textbf{x}}    \\ 
 \ \ \   \overleftarrow{\textbf{x}}  
\end{array}
\right]^{T}  P_{a} \left[
\begin{array}{l}  
 \overrightarrow{\textbf{x}}-\tilde{\textbf{x}}   \\ 
 \ \  \  \overleftarrow{\textbf{x}}  
\end{array}
\right],
\label{LyapunovFunction1}
\end{align}%
where $ P_{a} =\text{diag}\{\alpha I_{Nn}, \alpha I_{Nn}\}$, $ \alpha>0 $ is an adjustable constant, and the fact that $ (\overrightarrow{\textbf{x}})^{T}\overleftarrow{\textbf{x}}=0 $ and $ 1^{T}_{N} \mathcal{T}=0^{T}_{N}$ are used.

The time derivative of $ V_{a}(\textbf{x}) $ is expressed as
\vspace*{-1pt}
\begin{equation}
\dot{V}_{a}(\textbf{x})= - \alpha  (\textbf{x}-\tilde{\textbf{x}})^{T} [\mathcal{R}^{T}  \textbf{F} (\textbf{x}) + \kappa (\mathcal{L}\otimes I_{n})\textbf{x} ]. \label{L1}
\end{equation} 

As $ 0_{Nn}=-\mathcal{R}^{T} \textbf{F}(\tilde{\textbf{x}}) - \kappa (\mathcal{L}\otimes I_{n}) \tilde{\textbf{x}} $, (\ref{L1}) can be expressed as 
\vspace*{-1pt}
\begin{equation}
\dot{V}_{a}(\textbf{x})= - \alpha   (\textbf{x}-\tilde{\textbf{x}})^{T} [\mathcal{R}^{T}  (\textbf{F} (\textbf{x})-\textbf{F}(\tilde{\textbf{x}})) + \kappa (\mathcal{L}\otimes I_{n}) (\textbf{x}-\tilde{\textbf{x}}) ]. \label{L2}
\end{equation} 

In light of $ \tilde{\textbf{x}}=1_{N}\otimes x^{*} $, $ \textbf{x}=\overrightarrow{\textbf{x}}+ \overleftarrow{\textbf{x}} $, and $ (\mathcal{L} \otimes I_{n}) \overrightarrow{\textbf{x}}= 0_{Nn} $ under Assumption \ref{AssumptionGraph}, then the first term in 
(\ref{L2}) becomes
\vspace*{-3pt}
\begin{align}
&-   (\textbf{x}-\tilde{\textbf{x}})^{T} \mathcal{R}^{T} [  \textbf{F} (\textbf{x})-\textbf{F}(\tilde{\textbf{x}})]   \label{L3} \\
&= -(\overleftarrow{\textbf{x}})^{T} \mathcal{R}^{T}  [ \textbf{F} (\textbf{x})-\textbf{F}(\overrightarrow{\textbf{x}}) ]  - (\overleftarrow{\textbf{x}})^{T} \mathcal{R}^{T}[  \textbf{F}(\overrightarrow{\textbf{x}})-\textbf{F}(\tilde{\textbf{x}})]  \notag \\
&  \ \ \ - ( \overrightarrow{\textbf{x}} - \tilde{\textbf{x}})^{T} \mathcal{R}^{T} [ \textbf{F} (\textbf{x})-\textbf{F}(\overrightarrow{\textbf{x}})]   -  ( \overrightarrow{\textbf{x}} - \tilde{\textbf{x}})^{T} \mathcal{R}^{T} [  \textbf{F}(\overrightarrow{\textbf{x}})-\textbf{F}(\tilde{\textbf{x}})]. \notag 
\end{align}

It follows from Assumption \ref{AssumptionSmooth} that according to the $ \iota_{F} $-Lipschitz continuity of $ F $, it yields that $ \|F(\textbf{x})-F(\overrightarrow{\textbf{x}})  \| \leq \iota_{F} \| \overleftarrow{\textbf{x}} \| $. Further, $ \| \textbf{F} (\textbf{x})-\textbf{F}(\overrightarrow{\textbf{x}})  \| \leq \iota_{\textbf{F}} \| \overleftarrow{\textbf{x}} \| $ for certain scalar $ \iota_{\textbf{F}}>0 $. In addition, since $ \| \mathcal{R}^{T} \|=1 $, $ \textbf{F}(\overrightarrow{\textbf{x}})=F(\bar{\textbf{x}})$,  and $  \textbf{F}(\tilde{\textbf{x}})=F(x^{*})=0 $, 
\vspace*{-2pt}
\begin{align}
 - (\overleftarrow{\textbf{x}})^{T} \mathcal{R}^{T}[  \textbf{F}(\overrightarrow{\textbf{x}})-\textbf{F}(\tilde{\textbf{x}})] & = - (\overleftarrow{\textbf{x}})^{T} \mathcal{R}^{T} (F(\bar{\textbf{x}}) - F(x^{*})) \notag \\
 & \leq \iota_{F} \| \overleftarrow{\textbf{x}} \| \| \bar{\textbf{x}}- x^{*}\|, \label{L4}
\end{align}
\vspace*{-20pt}
\begin{align}
- ( \overrightarrow{\textbf{x}} - \tilde{\textbf{x}})^{T} \mathcal{R}^{T} [ \textbf{F} (\textbf{x})-\textbf{F}(\overrightarrow{\textbf{x}})] &= - (\bar{\textbf{x}}-x^{*})^{T} (\textbf{F} (\textbf{x})-\textbf{F}(\overrightarrow{\textbf{x}})) \notag \\
& \leq \iota_{\textbf{F}} \|\bar{\textbf{x}}-x^{*} \| \| \overleftarrow{\textbf{x}}\|, \label{L5}
\end{align}
where the fact that $ \mathcal{R} \overrightarrow{\textbf{x}}=\bar{\textbf{x}} $ and $ \mathcal{R} \tilde{\textbf{x}}=x^{*} $ is used, and exploiting the $ \varepsilon $-strong monotonicity of $ F $, we can obtain 
\vspace*{-2pt}
\begin{align}
& -  ( \overrightarrow{\textbf{x}} - \tilde{\textbf{x}})^{T} \mathcal{R}^{T} [  \textbf{F}(\overrightarrow{\textbf{x}})-\textbf{F}(\tilde{\textbf{x}})] \notag \\
&  =- (\bar{\textbf{x}}-x^{*})^{T}(F(\bar{\textbf{x}})-F(x^{*})) \leq -\varepsilon \| \bar{\textbf{x}}-x^{*} \|^{2}. \label{L6}
\end{align} 

In addition, the second term in (\ref{L2}) can be rewritten as
\vspace*{-3pt}
\begin{align}
&-(\textbf{x}-\tilde{\textbf{x}})^{T}(\mathcal{L} \otimes I_{n}) (\textbf{x}-\tilde{\textbf{x}})  = -(\overrightarrow{\textbf{x}}+\overleftarrow{\textbf{x}} )^{T}(\mathcal{L}\otimes I_{n}) (\overrightarrow{\textbf{x}}+\overleftarrow{\textbf{x}})     \notag \\
& = -\frac{1}{2} \overleftarrow{\textbf{x}}^{T} [(\mathcal{L}^{T}+\mathcal{L}) \otimes I_{n}] \overleftarrow{\textbf{x}}  \leq - \lambda_{2} (\hat{\mathcal{L}}) \| \overleftarrow{\textbf{x}} \|^{2},
 \label{L7}
 \end{align} 
where $ \lambda_{2} (\hat{\mathcal{L}}) $ is a minimal positive eigenvalue of  $ \hat{\mathcal{L}}=\frac{1}{2}( \mathcal{L}^{T}+\mathcal{L})$.  

\vspace*{3pt}
Let $ \iota=\max\{\iota_{F},\iota_{\textbf{F}}\} $. Substituting (\ref{L4})-(\ref{L7}) into (\ref{L2}) gives 
\vspace*{-2pt}
\begin{align}
\hspace{-0.3em}
\dot{V}_{a}(\textbf{x}) & \leq - \alpha  \lbrace (\kappa \lambda_{2} ( \hat{\mathcal{L}}) - \iota ) \| \overleftarrow{\textbf{x}} \|^{2}  + \varepsilon \|\bar{\textbf{x}} -x^{*} \|^{2} - 2 \iota \| \overleftarrow{\textbf{x}} \| \| \bar{\textbf{x}} -x^{*} \|  \rbrace  \notag \\ 
& = - \alpha \left[
\begin{array}{l}  
\|\bar{\textbf{x}}-x^{*}\|  \\ 
 \ \ \|\overleftarrow{\textbf{x}} \|
\end{array}
\right]^{T}      
\left[
\begin{array}{cc} 
\varepsilon & - \iota \\
-\iota  & \kappa \lambda_{2} ( \hat{\mathcal{L}}) - \iota 
\end{array}
\right]
 \left[
\begin{array}{l}  
\| \bar{\textbf{x}}-x^{*} \|  \\ 
 \ \  \|\overleftarrow{\textbf{x}} \|
\end{array}
\right], \notag \\
& = - \frac{1}{2} \left[
\begin{array}{l}  
\|\overrightarrow{\textbf{x}}-\tilde{\textbf{x}} \| \\ 
 \ \  \|\overleftarrow{\textbf{x}} \|
\end{array}
\right]^{T}     
Q_{a}
\left[
\begin{array}{l}  
\|\overrightarrow{\textbf{x}}-\tilde{\textbf{x}} \| \\ 
 \ \  \|\overleftarrow{\textbf{x}} \|
\end{array}
\right],
 \label{L8}
\end{align} 
where the fact that $ \|\bar{\textbf{x}}-x^{*}\|=\frac{1}{\sqrt{N}} \|\overrightarrow{\textbf{x}}-\tilde{\textbf{x}} \| $ is used and $ Q_{a}= 2\alpha  \\ \left[
\begin{array}{cc} 
\frac{\varepsilon}{N} & - \frac{\iota}{\sqrt{N}} \\
-\frac{\iota}{\sqrt{N}} & \kappa \lambda_{2} ( \hat{\mathcal{L}}) - \iota 
\end{array}
\right] >0 $ if $ \kappa > \frac{1}{\lambda_{2} (\hat{\mathcal{L}}) } (\frac{\iota^{2}}{\varepsilon} +\iota) $. 

\vspace{2pt}
Thus, it concludes that there exists positive scalars $\alpha, \epsilon, \iota $, and a matrix $ P_{a} $, so that for the Lyapunov function in (\ref{LyapunovFunction1}), 
\vspace{-5pt}
\begin{equation}
\hspace{-0.2em}
V_{a}(\textbf{x})=\frac{1}{2} \chi^{T} P_{a}  \chi    \Rightarrow   \dot{V}_{a} (\textbf{x}) \leq - \lambda_{a} V_{a}(\textbf{x}),  \   t \in [t_{2k}, t_{2k+1}), \label{L9}
\end{equation}   
where $ \chi= \text{col} ( \overrightarrow{\textbf{x}}-\tilde{\textbf{x}} , \overleftarrow{\textbf{x}} ) $ and $ \lambda_{a}=\lambda_{min}(Q_{a})/\lambda_{max}(P_{a})>0 $.

\vspace*{2pt}
\textbf{Step ii):} when considering the presence of DoS attacks during $[t_{2k+1}, t_{2(k+1)}) $, we select the Lyapunov function  as
\vspace*{-5pt}
\begin{align}
V_{b}(\textbf{x})&= \frac{\beta}{2} (\textbf{x}-\tilde{\textbf{x}})^{T}(\textbf{x}-\tilde{\textbf{x}})  
=\frac{1}{2}  \chi^{T}
P_{b} 
 \chi,
\label{LyapunovFunction2}
\end{align}%
where $ P_{b} =\text{diag}\{\beta I_{Nn}, \beta I_{Nn}\} $, $ \beta \neq \alpha $ is a positive constant.

\vspace{2pt}
In the presence of DoS attacks, $ \dot{\textbf{x}}= -\mathcal{R}^{T} \textbf{ F}(\textbf{x})-(\mathcal{L}^{\Psi(t)} \otimes I_{n})\textbf{x}$ by (\ref{NEClosedCompact1}). Since $ x^{*} $ is the unique NE, $ 0_{Nn}= -\mathcal{R}^{T} \textbf{ F}(\tilde{\textbf{x}})-(\mathcal{L}^{\Psi(t)} \otimes I_{n})\tilde{\textbf{x}}$. 
By adding this equation, it follows from (\ref{L3})-(\ref{L8}) that the time derivative of $ V_{b}(\textbf{x}) $ can be described by   
\vspace{-3pt}
\begin{align}
\dot{V}_{b}(\textbf{x}) &= - \beta   (\textbf{x}-\tilde{\textbf{x}})^{T} [\mathcal{R}^{T}  (\textbf{F} (\textbf{x})-\textbf{F}(\tilde{\textbf{x}})) +(\mathcal{L}^{\Psi(t)} \otimes I_{n}) (\textbf{x}-\tilde{\textbf{x}}) ]   \notag \\
& \leq - \beta \left[
\begin{array}{l}  
\|\bar{\textbf{x}}-x^{*}\|  \\ 
\ \ \|\overleftarrow{\textbf{x}} \|
\end{array}
\right]^{T}      
\left[
\begin{array}{cc} 
\varepsilon & - \iota \\
-\iota  &   - \iota 
\end{array}
\right]
\left[
\begin{array}{l}  
\| \bar{\textbf{x}}-x^{*} \|  \\ 
\ \  \|\overleftarrow{\textbf{x}} \|
\end{array}
\right] \notag \\
& \ \ \ - \beta   ( \overrightarrow{\textbf{x}} -\tilde{\textbf{x}} +\overleftarrow{\textbf{x}} )^{T} (\mathcal{L}^{\Psi(t)} \otimes I_{n}) (\overrightarrow{\textbf{x}} -\tilde{\textbf{x}} +\overleftarrow{\textbf{x}} ) \notag \\
& \leq \beta  \left[
\begin{array}{l}  
\|\overrightarrow{\textbf{x}}-\tilde{\textbf{x}} \| \\ 
\ \  \|\overleftarrow{\textbf{x}} \|
\end{array}
\right]^{T}     
 \left[
\begin{array}{cc} 
-\frac{\varepsilon}{N} &   \frac{\iota}{\sqrt{N}} \\
\frac{\iota}{\sqrt{N}} & \iota 
\end{array}
\right] 
\left[
\begin{array}{l}  
\|\overrightarrow{\textbf{x}}-\tilde{\textbf{x}} \| \\ 
\ \  \|\overleftarrow{\textbf{x}} \|
\end{array}
\right] \notag \\
& \ \ \ + \beta \|\mathcal{L}^{\Psi(t)} \|  \left[
\begin{array}{l}  
\|\overrightarrow{\textbf{x}}-\tilde{\textbf{x}} \| \\ 
\ \  \|\overleftarrow{\textbf{x}} \|
\end{array}
\right]^{T}     
\left[
\begin{array}{l}  
\|\overrightarrow{\textbf{x}}-\tilde{\textbf{x}} \| \\ 
\ \  \|\overleftarrow{\textbf{x}} \|
\end{array}
\right] \notag \\
& = \frac{1}{2} \left[
\begin{array}{l}  
\|\overrightarrow{\textbf{x}}-\tilde{\textbf{x}} \| \\   
\ \  \|\overleftarrow{\textbf{x}} \|
\end{array}
\right]^{T} Q_{b}     
\left[
\begin{array}{l}  
\|\overrightarrow{\textbf{x}}-\tilde{\textbf{x}} \| \\ 
\ \  \|\overleftarrow{\textbf{x}} \|
\end{array}
\right]
\label{LL1}
\end{align}
where $ Q_{b}= 2\beta  \left[
\begin{array}{cc} 
c-\frac{\varepsilon}{N} &   c+\frac{\iota}{\sqrt{N}} \\
c+ \frac{\iota}{\sqrt{N}} & c+\iota 
\end{array}
\right] $ with $ c= \|\mathcal{L}^{\Psi(t)} \| $, may have both positive and negative eigenvalues. 

Thus, it concludes that there exists positive scalars $\beta, \epsilon, \iota $, and a matrix $ P_{b} $, so that for the Lyapunov function in (\ref{LyapunovFunction2}), 
\vspace{-5pt}
\begin{equation}
\hspace{-0.3em}
V_{b}(\textbf{x})=\frac{1}{2} \chi^{T} P_{b}  \chi    \Rightarrow  \dot{V}_{b} (\textbf{x}) \leq \lambda_{b} V_{b}(\textbf{x}),  \ t \in [t_{2k+1},t_{2k+2}), \label{LL3}
\end{equation}   
where  $ \lambda_{a}=\sigma _{\max}(Q_{b})/\lambda_{min}(P_{b})>0 $ and $ \sigma _{\max}(Q_{b}) $ denotes the maximum singular value of $ Q_{b} $. 

\vspace{2pt}
\textbf{Step iii):} we analyze the exponential convergence of the closed-loop system from a switching perspective \cite{Hu15IFAC,Hu15IJNRC,Hu17ACC,Hu15Cyber,Hu19TCST}. 

Let $ \delta(t) \in \{a,b\} $ be a switching signal. Then, we can choose 
\vspace{-5pt}
\begin{equation}
V(t)=\left\{ 
\begin{array}{c} 
\hspace{-0.8em} 
V_{a}(\textbf{x}) ,  \  \text{if} \ t \in  [t_{2k},t_{2k+1}) , \\ 
 V_{b}(\textbf{x}) ,   \text{ if} \ t \in   [t_{2k+1},t_{2k+2}),
\end{array}%
\right. 
\label{VV5}
\end{equation}
where $ V_{a}(\textbf{x}) $ and $ V_{b}(\textbf{x}) $ are defined in (\ref{LyapunovFunction1}) and (\ref{LyapunovFunction2}), respectively.
 
Suppose that $ V_{a} $ is activated in $ [t_{2k},t_{2k+1}) $, while $ V_{b} $ is activated in $ [t_{2k+1},t_{2k+2}) $. Then, by (\ref{L9}) and (\ref{LL3}), we have 
\vspace{-5pt}
\begin{equation}
V(t) \leq \left\{ 
\begin{array}{c}
\hspace{-0.5em}  
e^{-\lambda_{a}(t-t_{2k}) } V_{a}(t_{2k}) , \  \text{if} \ t \in \ [t_{2k},t_{2k+1}) , \\ 
\hspace{-0.5em}  
e^{\lambda_{b}(t-t_{2k+1}) } V_{b}(t_{2k+1}) , \  \text{ if} \ t \in \ [t_{2k+1},t_{2k+2}).
\end{array}%
\right. 
\label{VV6}
\end{equation}

The closed-loop system is switched at $ t= t^{+}_{2k} $ or $ t= t^{+}_{2k+1} $. Let $ \mu=\max\{\lambda_{max}(P_{a})/\lambda_{min}(P_{b}),\lambda_{max}(P_{b})/\lambda_{min}(P_{a})\} > 1$, and next, we discuss the following two cases: 

\vspace*{1pt}
Case a): If $ t \in   [t_{2k},t_{2k+1}) $, it follows from (\ref{VV6}) that 
\begin{align}
V(t) & \leq e^{-\lambda_{a} (t-t_{2k})}V_{a}(t_{2k})\leq \mu e^{-\lambda_{a}
	(t-t_{2k})}V_{b}(t_{2k}^{-})  \notag \\
&\leq \mu e^{-\lambda_{a} (t-t_{2k})}[e^{\lambda_{b}
	(t_{2k}-t_{2k-1})}V_{b}(t_{2k-1})]  \notag \\
&\leq \mu e^{-\lambda_{a} (t-t_{2k})}e^{\lambda_{b} (t_{2k}-t_{2k-1})} [\mu
V_{a}(t_{2k-1}^{-})]  \notag \\
&=\mu ^{2}e^{-\lambda_{a} (t-t_{2k})}e^{\lambda_{b}
	(t_{2k}-t_{2k-1})}V_{a}(t_{2k-1}^{-}) \notag \\
&\leq \mu ^{2}e^{-\lambda_{a} (t-t_{2k})}e^{\lambda_{b} (t_{2k}-t_{2k-1})}  \lbrack e^{-\lambda_{a} (t_{2k-1}-t_{2k-2})}V_{a}(t_{2k-2})] 
\notag \\
&\leq \cdots  \notag \\
&\leq \mu^{2k} e^{-\lambda_{a} | \Xi_{s} (t_{0}, t)| }e^{\lambda_{b}
| \Xi_{a} (t_{0}, t)| }V_{a}(t_{0}).  \label{VV7}
\end{align}

Case b): If $ t \in   [t_{2k+1},t_{2(k+1)}) $, it follows from (\ref{VV6}) that 
\begin{align}
V(t) &\leq e^{\lambda_{b} (t-t_{2k+1})} V_{b}(t_{2k+1}) \leq \mu e^{\lambda_{b} (t-t_{2k+1})}V_{a}(t_{2k+1}^{-}) \notag \\
&\leq \mu e^{\lambda_{b} (t-t_{2k+1})}[e^{-\lambda_{a}(t_{2k+1}-t_{2k})} V_{a}(t_{2k}) \notag \\ 
& \leq \mu e^{\lambda_{b}(t-t_{2k+1})} e^{-\lambda_{a}(t_{2k+1}-t_{2k})} [  \mu V_{b}(t^{-}_{2k-1})  ] \notag \\ 
& \leq \mu^{2} e^{\lambda_{b}(t-t_{2k+1})} e^{-\lambda_{a}(t_{2k+1}-t_{2k}) }  [ e^{\lambda_{b}(t_{2k-1}-t_{2k-2})} V_{b}(t_{2k-2})  ]   \notag \\ 
&\leq \cdots  \notag \\
&\leq \mu ^{2k+1} e^{-\lambda_{a} | \Xi_{s} (t_{0}, t)| }e^{\lambda_{b} | \Xi_{a} (t_{0}, t)| } V_{a}(t_{0}).  \label{VV8}
\end{align} 
 
\textbf{Step iv):} we consider bounds on attack frequency and duration.

\vspace*{2pt}
Notice that $N_{a}(t_{0},t)= k$ for $t \in [t_{2k},t_{2k+1})$ and $k+1$ for  $t \in [t_{2k+1},t_{2(k+1)})$. Thus, for $\forall t\geq t_{0},$ by (\ref{VV7}) and (\ref{VV8}),    
\vspace*{-1pt}
\begin{equation}
V(t)\leq \mu^{2N_{a}(t_{0},t)}e^{-\lambda_{a} |\Xi_{s}(t_{0},t)|}e^{\lambda_{b}|\Xi_{a}(t_{0},t)|}V(t_{0}).
\label{VV9}
\end{equation}

Notice that for all $t\geq t_{0},$ $|\Xi_{s}(t_{0},t)|=t-t_{0}-|\Xi_{a}(t_{0},t)|$ and $|\Xi_{a}(t_{0},t)|\leq T_{0}+(t-t_{0})/T_{a} $ by Definition \ref{Attack Duration}. Thus, we have 
\vspace*{-3pt}
\begin{align}
&-\lambda_{a}(t-t_{0}-|\Xi_{a}(t_{0},t)|)+\lambda_{b}|\Xi_{a}(t_{0},t)|  \notag \\
&= -\lambda_{a}(t-t_{0})+(\lambda_{a}+\lambda_{b})  |\Xi_{a}(t_{0},t)|  \notag \\
&\leq -\lambda_{a}(t-t_{0})+(\lambda_{a}+\lambda_{b}) [T_{0}+(t-t_{0})/T_{a}].  \label{VV10}
\end{align}

Next, substituting (\ref{VV10}) into (\ref{VV9}) yields
\vspace*{-2pt}
\begin{align}
V(t) &\leq \mu ^{2N_{a}(t_{0},t)}e^{-\lambda_{a}(t-t_{0}-| \Xi_{a}(t_{0},t)|)}e^{\lambda_{b}|\Xi_{a}(t_{0},t)|}V(t_{0})  \label{VV11} \\
&\leq e^{(\lambda_{a}+\lambda_{b})T_{0}}e^{-\lambda_{a}(t-t_{0})}e^{\frac{(\lambda_{a}+\lambda_{b})}{\tau_{a}}(t-t_{0})} e^{2\ln (\mu ) N_{a}(t_{0},t)}V(t_{0}).  \notag
\end{align}

By exploiting the attack condition in (\ref{Condition1}), we can have 
\vspace*{-3pt}
\begin{equation}
2\ln(\mu) N_{a}(t_{0},t) \leq  2\ln(\mu) N_{0}+ \eta^{*}(t-t_{0}) .   
\label{VV111}
\end{equation} 

Let $\eta=\lambda_{a}-(\lambda _{a}+\lambda _{b})/T_{a}-\eta^{\ast}>0$. 
Based on another attack condition in (\ref{Condition2}), and using (\ref{VV111}), we can rewrite (\ref{VV11}) as
\vspace*{-2pt}
\begin{equation}
V(t)\leq e^{( \lambda_{a}+ \lambda_{b}) T_{0} +2 \ln(\mu) N_{0}} \ e^{-\eta
(t-t_{0})}V(t_{0}).  
\label{VV12}
\end{equation} 

Further, it follows from (\ref{LyapunovFunction1}), (\ref{LyapunovFunction2}),  and (\ref{VV12}) that 
\vspace*{-3pt}
\begin{equation}
|| \chi (t)||^{2} \leq \varsigma e^{-\eta
	(t-t_{0})} ||\chi(t_{0})||^{2},   
\label{VV13}
\end{equation} 
where $ \varsigma= e^{(\lambda_{a}+ \lambda_{b})T_{0} +2 \ln(\mu) } \varsigma_{a}/\varsigma_{b}$, $ \varsigma_{a}= \max\{  \lambda_{max}(P_{a}),\lambda_{max} ( \\ P_{b})\}$, and $ \varsigma_{b}=\min\{\lambda_{min}(P_{a}),\lambda_{min}(P_{b})\}$.

\vspace{3pt}
Therefore, it follows from (\ref{VV13}) that all estimate states $  \overrightarrow{\textbf{x}}-\tilde{\textbf{x}} $ and $\overleftarrow{\textbf{x}}$  are bounded, and converge to zero exponentially. Furthermore,  $ \lim_{t\rightarrow \infty} (\overrightarrow{\textbf{x}}-\tilde{\textbf{x}})=0_{Nn}$ and $ \lim_{t\rightarrow \infty} \overleftarrow{\textbf{x}} = 0_{Nn}$. Then, according to the coordinate transformation $ \overrightarrow{\textbf{x}} $, $ \overleftarrow{\textbf{x}} $ in (\ref{Xright}) and (\ref{Xleft}), and by using the fact that $ \textbf{x}= \overrightarrow{\textbf{x}}+ \overleftarrow{\textbf{x}}$, we can obtain that $\textbf{x} $ exponentially converges to $ \tilde{\textbf{x}} =1 \otimes x^{*}$. That is, Problem \ref{Problem} is solved. 
\end{proof}

\begin{remark}
Theorem 1 presented the main resilient distributed NE seeking result with an exponential convergence rate $\eta=\lambda_{a}-(\lambda _{a}+\lambda _{b})/T_{a}-\eta^{\ast}$. Here, $(1-\frac{1}{T_{a}}) \lambda _{a}$ is mainly used to measure the average rate of exponential decay of the stable subsystems, while $\lambda _{b}/T_{a} $ isy used to measure the exponential growth rate of unstable subsystems. The $ \eta^{*} $ explains the exponential growth due to switchings. Note that the convergence rate is not only affected by $ \lambda_{a} $ and $ \lambda_{b} $ that rely on the communication topology, number of players, and control gains, but the attack frequency and duration. Moreover, the larger values of attack frequency and duration are, the more active that those attacks are allowed to be.
\end{remark}

Notice that in the absence of DoS attacks, Theorem 1 can be reduced to the following corollary: 

\textbf{Corollary 1:} Under Assumptions 1-3, the following distributed NE seeking algorithm enables all players' estimated strategies to exponentially converge to the NE provided $ \kappa >  (\frac{\iota^{2}}{\varepsilon} +\iota)/\lambda_{2}(\hat{\mathcal{L}}) $.
\vspace{-6pt}
\begin{equation}
\dot{\textbf{x}}^{i}=-\mathcal{R}^{T}_{i} \triangledown_{i} J_{i}(\textbf{x}^{i})  + \textbf{e}^{i}, \  \textbf{e}^{i}=-\kappa\sum_{j=1}^{N}a_{ij} (\textbf{e}^{i}-\textbf{e}^{j})
, \ i\in \mathcal{V}.
\notag 
\end{equation}

\vspace{-5pt}
\begin{remark}
The corollary can cover some existing results (e.g., \cite{Ye17TAC,Pavel19TAC}) as special cases. Moreover, it can avoid restrictive graph coupling conditions in \cite{Pavel19TAC} by adding a proportional gain $ \kappa $, and remove the use of two-timescale singular perturbation that yields semi-global convergence in \cite{Ye17TAC,Pavel19TAC}. The design does not require any initial requirements and allow for a general directed graph.     
\end{remark}

\vspace{-5pt}  
\section{Numerical Simulations}
\vspace{-2pt}
In this section, numerical examples are presented to verify the effectiveness of the proposed NE seeking design. The communication graph for all players in examples is depicted in Fig. \ref{StronglyDigraphAttacks}, in which the original strongly connected digraph and the paralyzed graphs induced by attacks are presented, respectively.  

\begin{figure}[!h]
\centering
\includegraphics[width=6.0cm,height=3.6cm]{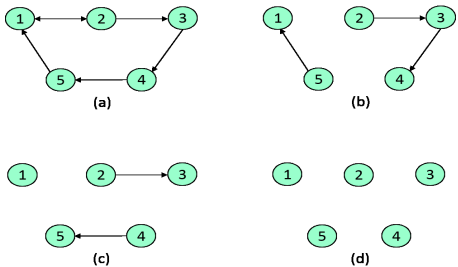}
\caption{Communication graph for the players in the examples: (a) original strongly connected digraph; and (b)-(d) paralyzed graphs under various DoS attacks.} 
\label{StronglyDigraphAttacks}
\end{figure}

\begin{example}
(Energy Consumption Game)
\end{example} 

In this example, we consider an energy consumption game of $ N $ players for Heating Ventilation and Air Conditioning (HVAC) system (see \cite{Ye17TAC}), where the cost function of each player $ i $ can be modeled by the following function:
\vspace{-5pt}
\begin{equation}
J_{i}(x_{i},x_{-i})=a_{i}(x_{i}-b_{i})^{2}+ \left(  c\sum_{j=1}^{N}x_{j} +d  \right)  x_{i}, \ i \in \mathcal{V}, \notag  
\end{equation}
where $ a_{i}>0, c>0 $, $ b_{i} $ and $ d $ are constants for $ i \in \mathcal{V} $. It can be verified that Assumptions \ref{AssumptionConvex} and \ref{AssumptionSmooth} are satisfied. Throughout this simulation, let $ a_{i}=1$, $ c=0.1 $, $ d=10 $ for each player. In the following simulation, we investigate the effectiveness of the proposed resilient distributed NE seeking algorithm in (\ref{NEController1})-(\ref{NE2}) from the perspective of network under DoS attacks with several comparisons with the existing algorithm, controller gain, attack frequency/duration, network topology and number of players. 

\vspace{-6pt}
\subsection{Resilient Algorithm for Exponential Distributed NE Seeking}
\vspace{-1pt}
We consider five players ($ N=5 $) in the game over a strongly connected digraph in the absence of attacks and three types of disconnected digraphs caused by DoS attacks as shown in Fig. \ref{StronglyDigraphAttacks}. Constants $ b_{i} $ for $ i =1,\cdots,5 $, are set to $ 10$, $15$, $20$, $25 $, and $30$, respectively. By certain calculation based on those parameters, the NE is $ x^{*} =\text{col} (2.0147,6.7766,11.5385,  16. 3004,21.0623)$  \cite{Ye17TAC}. 
The initial states are given by $ x^{i}_{i}(0)=\text{col}(-2,-4,-6, -8, -10)$ and $ x^{i}_{j}(0)=\text{col}(15,10,5,0)$, $\forall i \neq j $, which are not close to $ x^{*} $. The control gain of  algorithm in (\ref{NE2}) is set as $ \kappa=10 $ and the variable to balance the weight is $ \omega=\text{col}(\frac{1}{6},\frac{2}{6},\frac{1}{6},\frac{1}{6},\frac{1}{6}) $. 

Next, we perform the proposed resilient NE seeking algorithm in (\ref{NEController1})-(\ref{NE2}), and  simulation results are provided as shown in Fig.  \ref{Combination1}. In particular, Fig. \ref{Combination1}(a) shows the occurrence of DoS attacks, where attack frequency and duration conditions in Theorem 1 are satisfied. Fig. \ref{Combination1}(b) illustrates all players' estimate strategies on the NE $x^{*}$, while the relative errors of all players' actions $ \|\textbf{x}-\textbf{x}^{*} \| / \\ \|\textbf{x}^{*}\| $ are depicted in Fig. \ref{Combination1}(c). As observed, all players' estimate strategies reach consensus and converge to the NE exponentially. 

\begin{figure}[!t]
	\centering
	\hspace{-0.5em}
	\includegraphics[width=9.0cm,height=8.2cm]{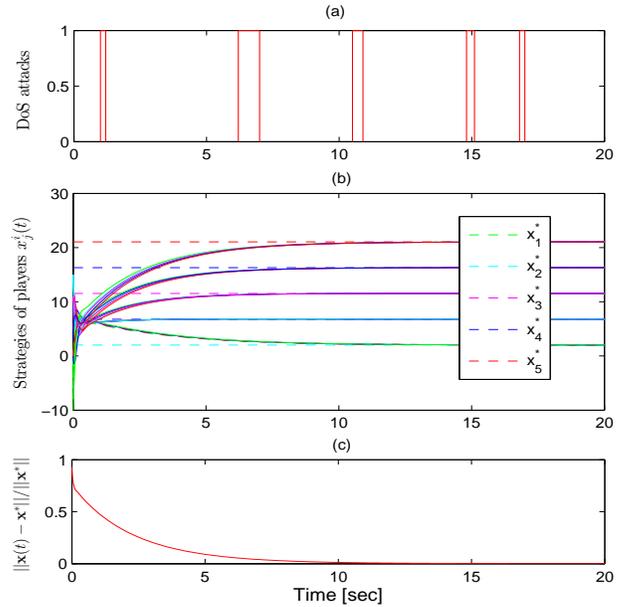}
	\caption{Simulated results of the proposed attack-resilient NE seeking algorithm in (\ref{NEController1})-(\ref{NE2}): (a) DoS attacks on various communication channels; (b) all players' estimated strategies $x^{i}_{j}(t), i, j\in \mathcal{V}$; and (c) relative errors of all players' actions.} 
	\label{Combination1}
\end{figure}

\vspace{-6pt}
\subsection{Algorithm Comparison}
\vspace{-1pt}
In order to make some comparisons, we perform the algorithm in \cite{Pavel19TAC} to further illustrate the proposed algorithm's effectiveness under DoS attacks. All simulation environments are set the same as those in the subsection V-A.  
It can be observed from Fig. \ref{ComparisionAlgorithm} (a) that in the presence of attacks, the design in \cite{Pavel19TAC}  cannot guarantee the exact convergence of all players' estimates to the NE $ x^{*} $. In contrast, Fig. \ref{ComparisionAlgorithm} (b) shows the performance of the proposed attack-resilient algorithm, which can verify the design's effectiveness. 

\begin{figure}[!t]
	\centering
	\begin{tabular}{c}
		%
		\hspace*{-1.0em}		
		\subfloat [Algorithm in \cite{Pavel19TAC} ] {\includegraphics[width=5.0cm,height=2.9cm]{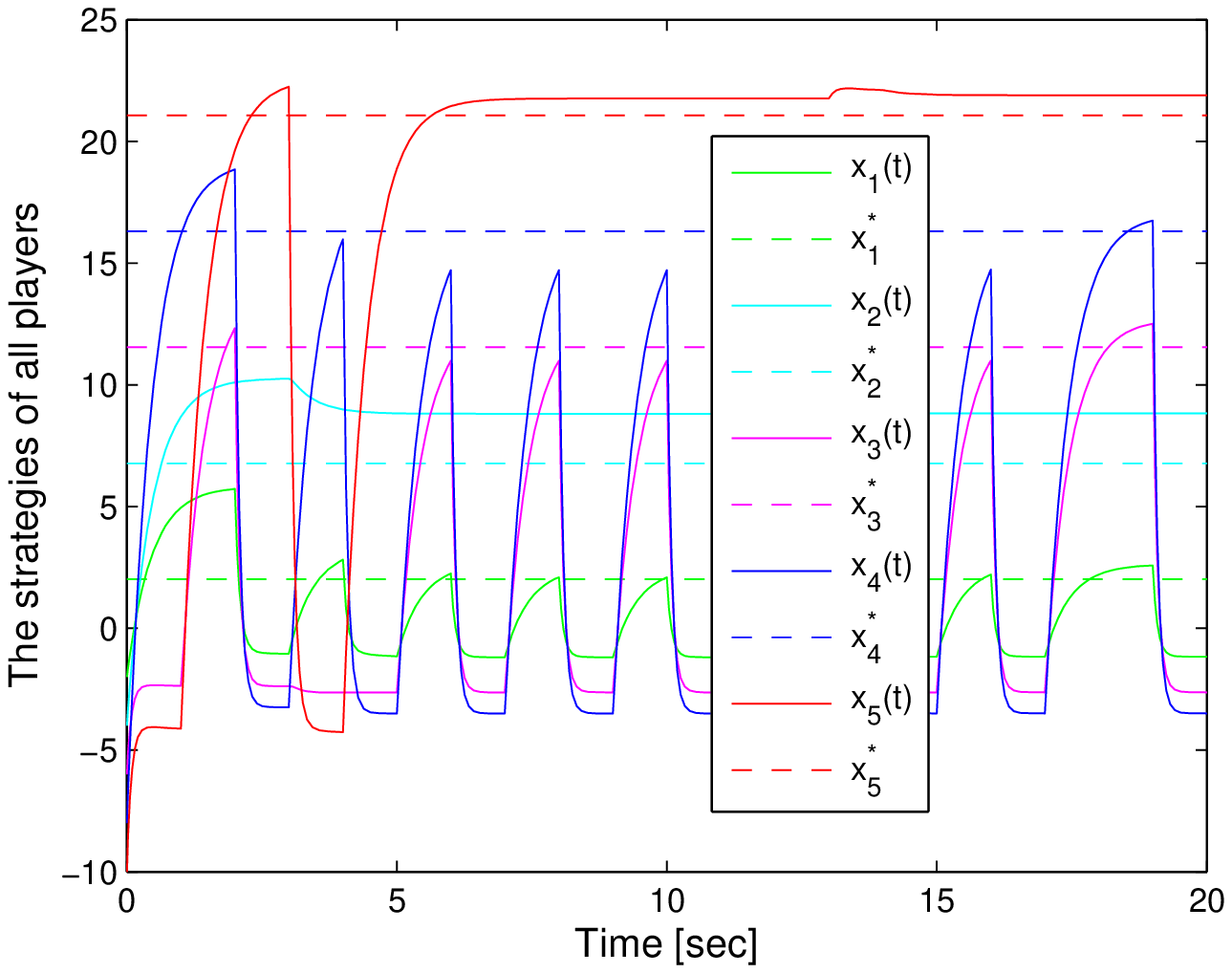} \label{A17}}
		
		\hspace*{-1.0em}
		\subfloat [Algorithm in (\ref{NEController1})-(\ref{NE2})]
		{\includegraphics[width=5.0cm,height=2.9cm]{StronglyDigraph_Attacks} \label{A} } 
	\end{tabular}	
	\caption{The plot of players' strategies $x_{i}(t)$ produced by the proposed algorithm and the algorithm in \cite{Pavel16TAC} in the presence of DoS attacks.
	}	
	\label{ComparisionAlgorithm} 
\end{figure}

\subsection{Performance Analysis of Algorithm}

\textit{\textbf{1) Controller gain:}} 
we show the influence of control gain $ \kappa $ on the performance of algorithm, we conduct the proposed algorithm with the same simulation setting in the subsection V-A, but under two different gains $ \kappa=1 $ and $ \kappa=5 $, respectively.  Fig. \ref{ControlGains} shows the plots of players' strategies and relative errors under different control gains, which implies that the larger the gain $\kappa$, the better the convergence performance is, which is as analyzed.

\textit{\textbf{2) Number of Player:}}
we increase the number of players to $ N=3$, $5$, $10 $ and analyze its influence on the performance of the proposed algorithm. We set $ b_{i}=5i+5 $ for each $ i=1,2,\cdots,N $, and select a cycle directed graph as the original communication graph. Fig. \ref{ComparisionDifferentplayers} depicts the performance of the proposed algorithm under the different number of players. The algorithm is scalable to various number of players, and the smaller the number of players, the better the performance is as expected.

\textit{\textbf{3) Network Topology:}} 
in this part, we investigate the influence of the  different original communication topologies on the performance of algorithm. As analyzed in the theorem, the larger the $ \lambda_{a} $, the better convergence performance is. Intuitively, the larger the nonzero eigenvalue of the Laplacian matrix, the larger the $ \lambda_{a} $. Fig. \ref{ComparisionDifferentgraphs} shows the logarithmic curve of relative errors. As we see, the performance is better for a graph with more links.

%


%

\textit{\textbf{4) Attack Frequency and Duration:}}
we illustrate the influence of attack frequency and duration on the performance of algorithm. According to the theorem, the frequency and duration of attacks have to be constrained to guarantee the convergence of all players' estimates to the NE. Simulated result is shown in Fig. \ref{ComparisionAttacks}, where a comparison, when the attack frequency and duration conditions in (\ref{Condition1}) and (\ref{Condition2}) hold and do not not, is provided. As can be seen, the result is consistent with the analysis in Remark 6.



\begin{figure}[!t]
	\centering
	\begin{tabular}{c}
		\hspace*{-1.2em}		
		\subfloat [$ \kappa=1 $]
		{\includegraphics[width=5.0cm,height=4.9cm]{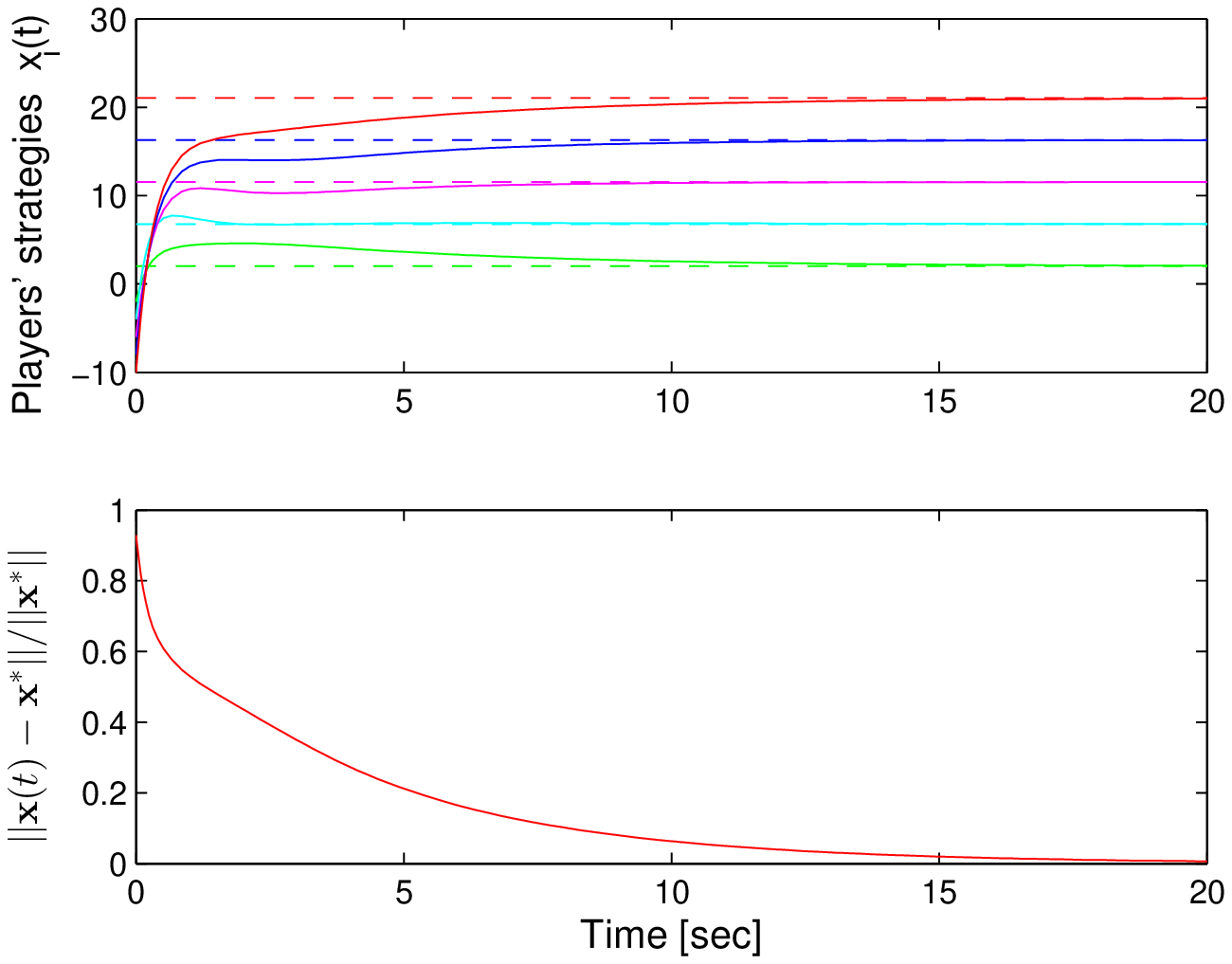} \label{k10} } 
		
		\hspace*{-1.2em} 
		\subfloat [$ \kappa=5 $] {\includegraphics[width=5.0cm,height=4.9cm]{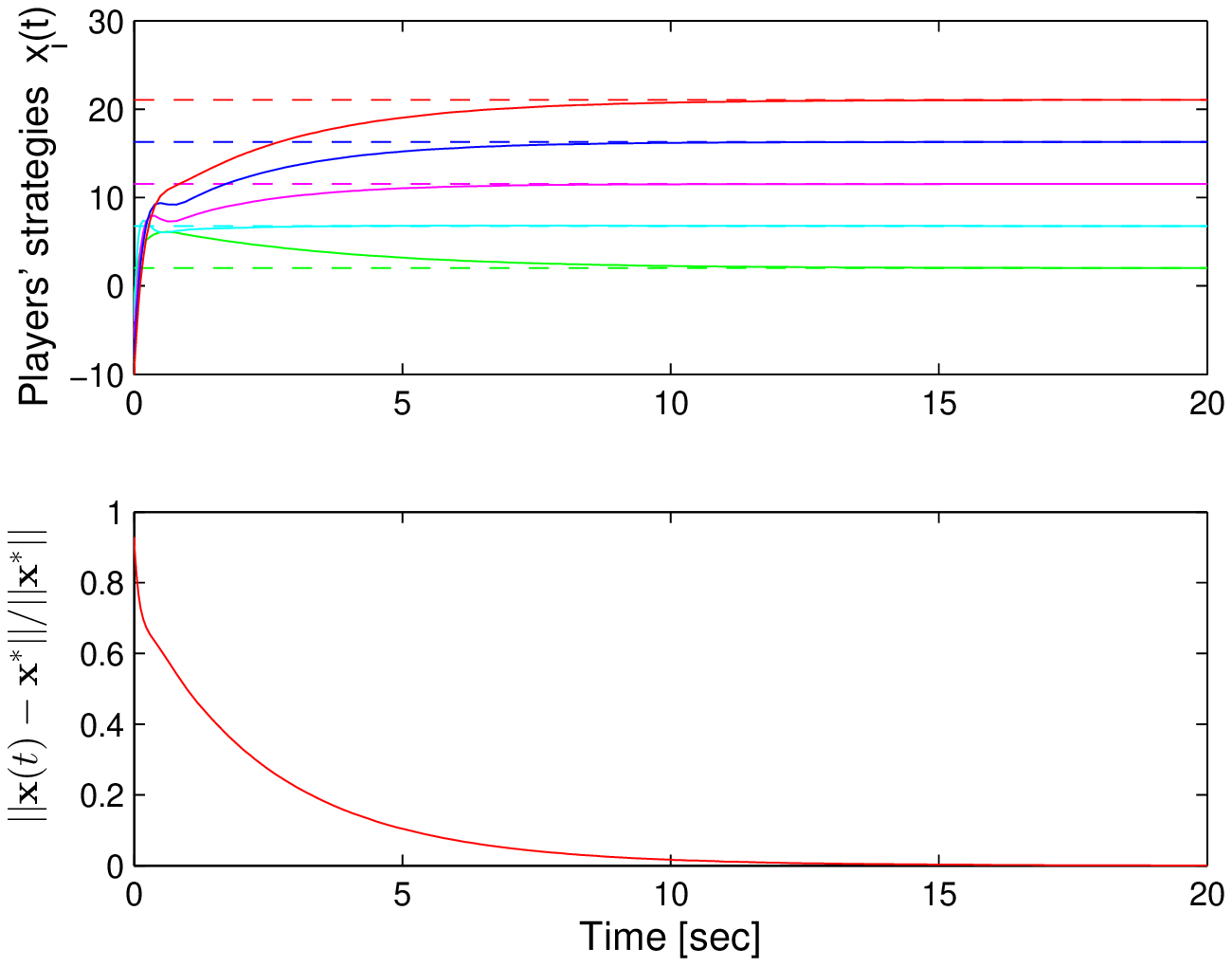} \label{k1}}
	\end{tabular}	
	\caption{The plot of players' strategies $x_{i}$ and their relative errors $ \|\textbf{x}-\textbf{x}^{*}\|/\|\textbf{x}^{*}\| $ produced by the proposed  algorithm in (\ref{NEController1})-(\ref{NE2}) with different controller gains. }	
	\label{ControlGains} 
\end{figure}

\begin{figure}[!t]
	\centering
	\begin{tabular}{c}
		\hspace*{-1.2em} 		
		\subfloat []
		{\includegraphics[width=6.0cm,height=1.0cm]{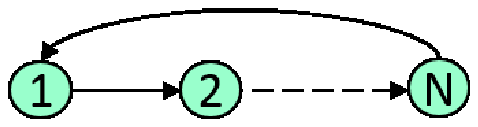} \label{DNgraph} }  
		
		\\
		
		\hspace*{-1.2em} 
		\subfloat [] {\includegraphics[width=4.8cm,height=3.1cm]{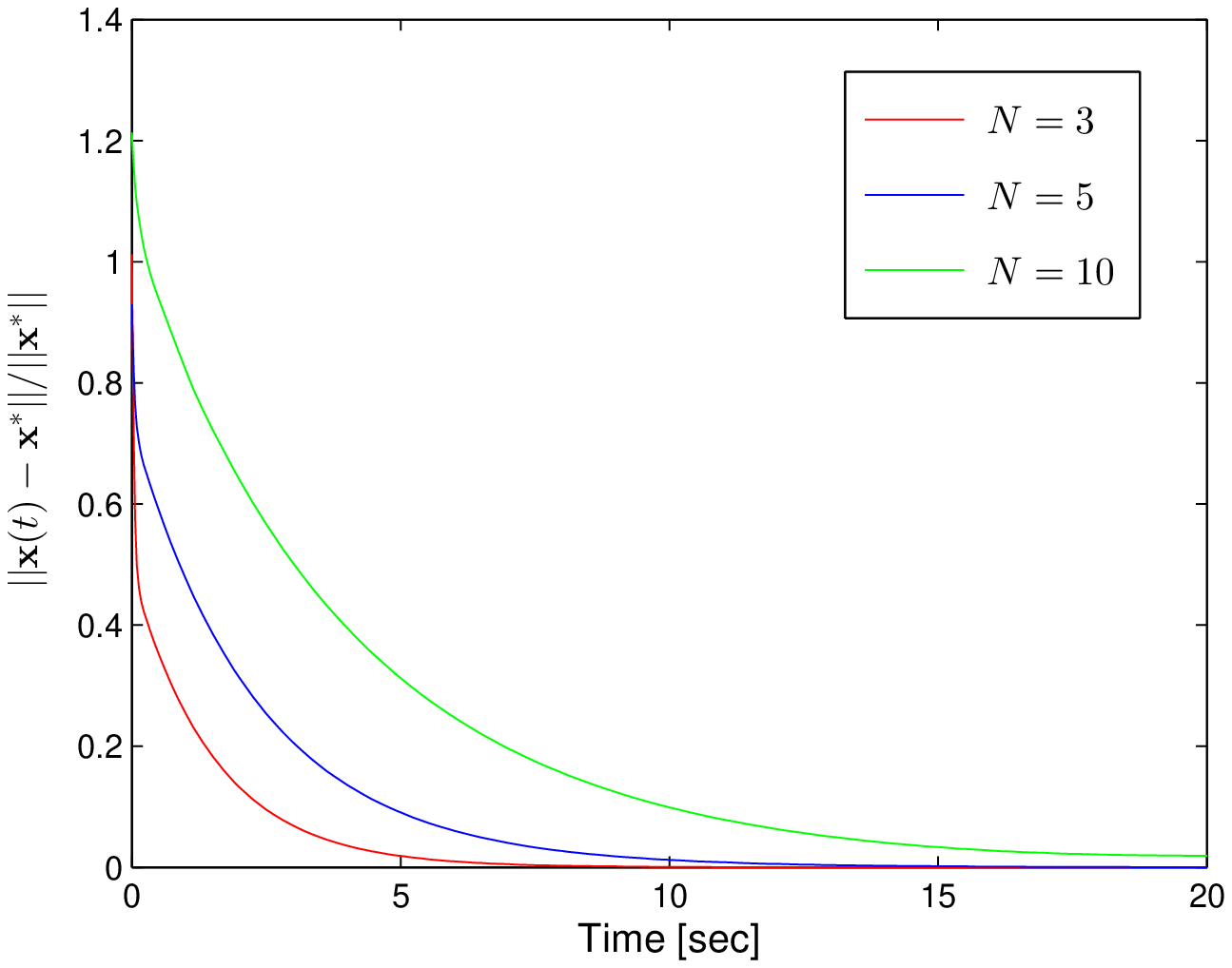} \label{Dn}}
		
		\subfloat []
		{\includegraphics[width=4.8cm,height=3.1cm]{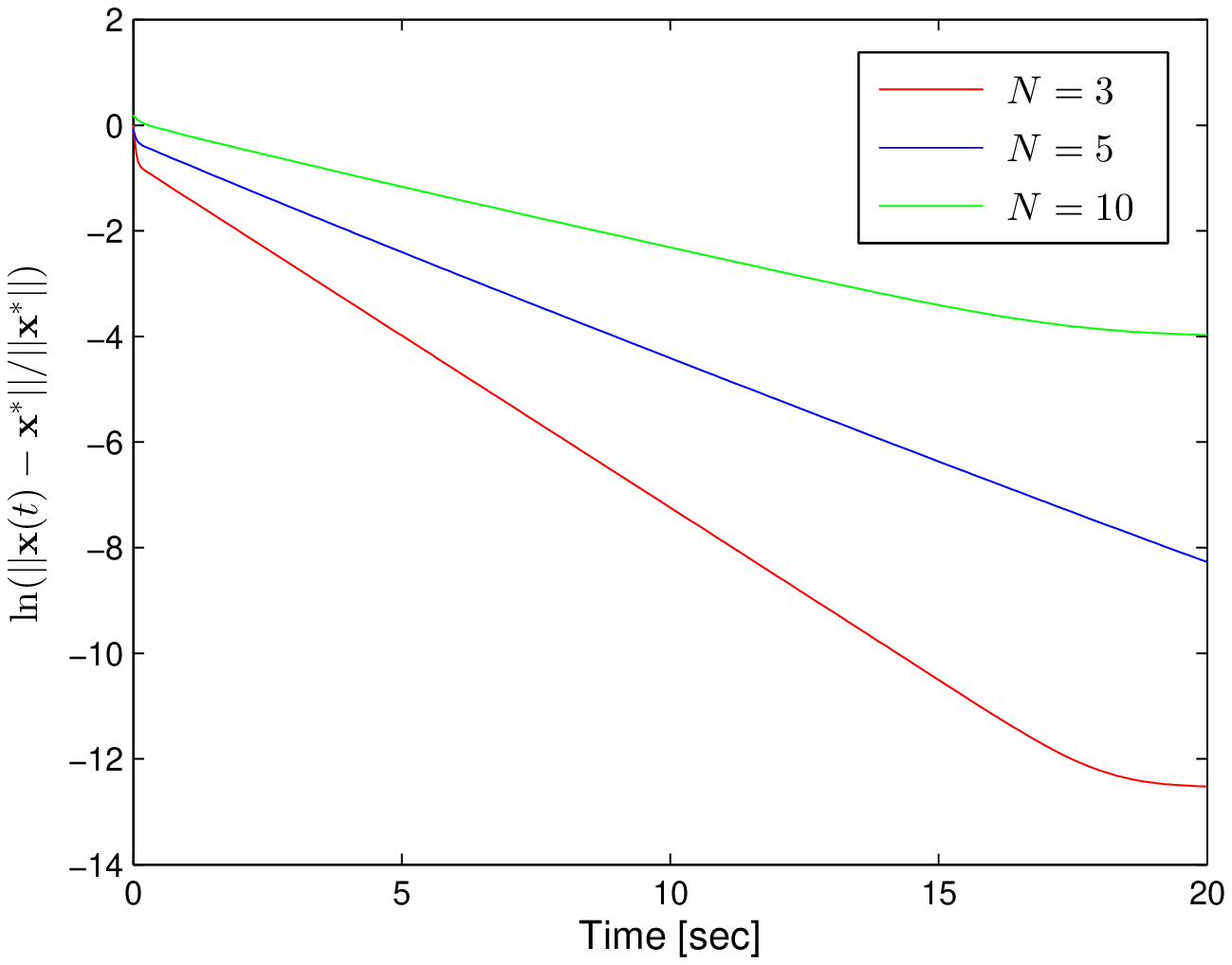} \label{DNLn} } 
		
	\end{tabular}	
	\caption{Performance of the proposed algorithm  (\ref{NEController1})-(\ref{NE2}) with different players: (a) the cycle digraph; (b) the relative error; and (c) its logarithmic curve. }	
	\label{ComparisionDifferentplayers} 
\end{figure}

\begin{figure}[!t]
	\centering
	\begin{tabular}{c}
		\hspace*{-1.2em}  
		\subfloat [] {\includegraphics[width=4.8cm,height=3.1cm]{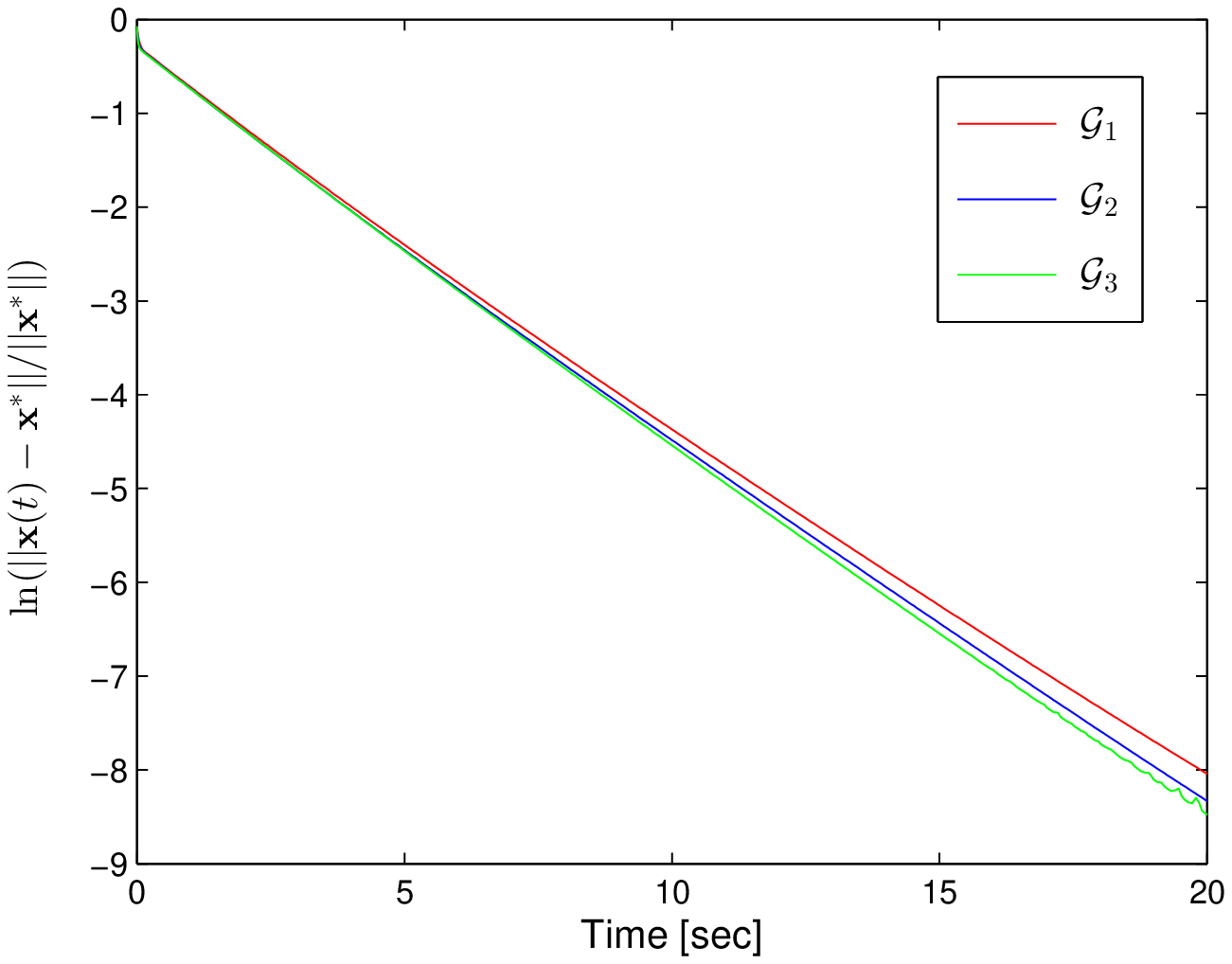} \label{Dcom}}
		\hspace*{-1.0em} 		
		\subfloat []
		{\includegraphics[width=5.0cm,height=2.8cm]{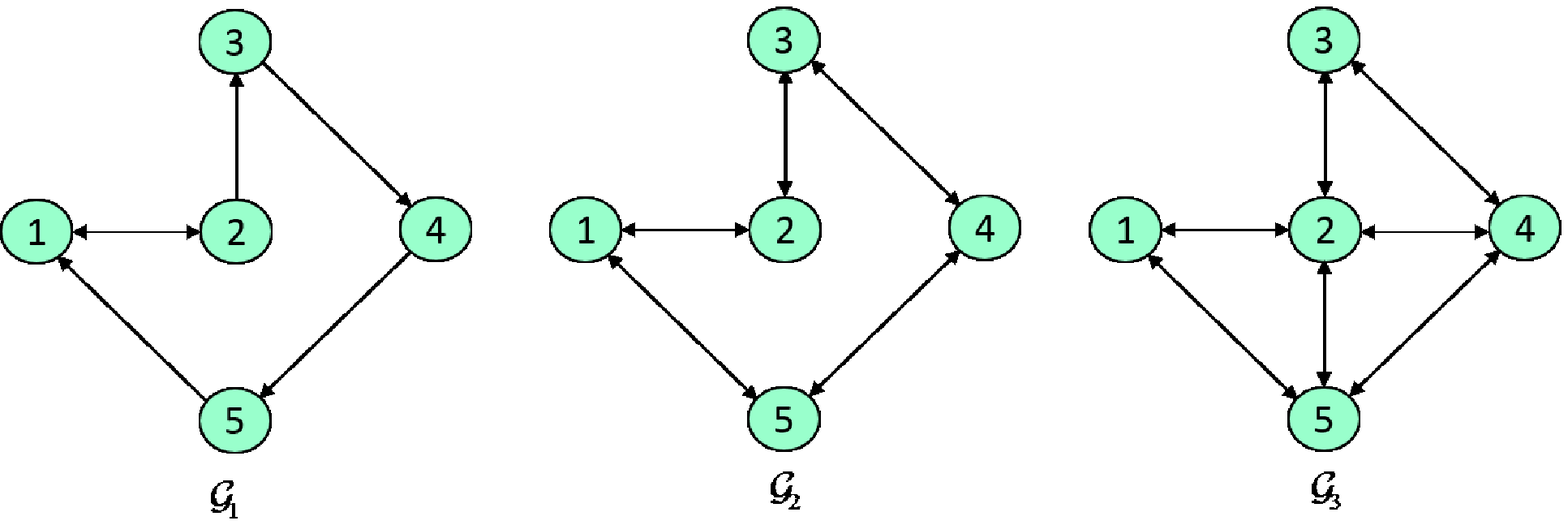} \label{Dgraph} }
	\end{tabular}	
	\caption{The logarithmic curve of the relative error produced by the proposed algorithm  (\ref{NEController1})-(\ref{NE2}) with different network graphs. }	%
	\label{ComparisionDifferentgraphs} 
\end{figure} 

\begin{figure}[!t]
	\centering
	\begin{tabular}{c}
		\hspace*{-1.2em}		
		\subfloat [Small attack requency and duration]
		{\includegraphics[width=4.8cm,height=3.0cm]{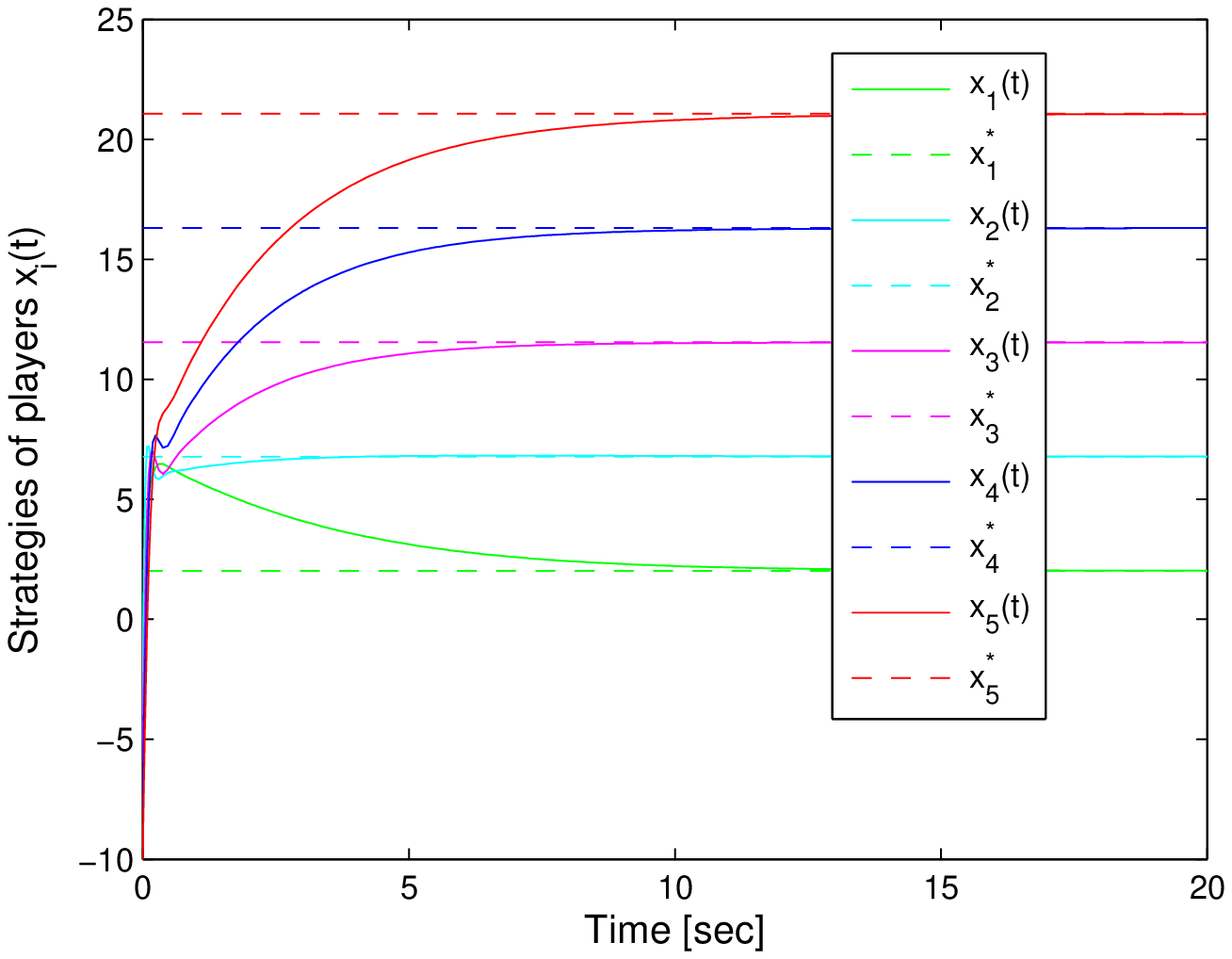} \label{A} } 
		
		\subfloat [Large attack requency and duration] {\includegraphics[width=4.8cm,height=3.0cm]{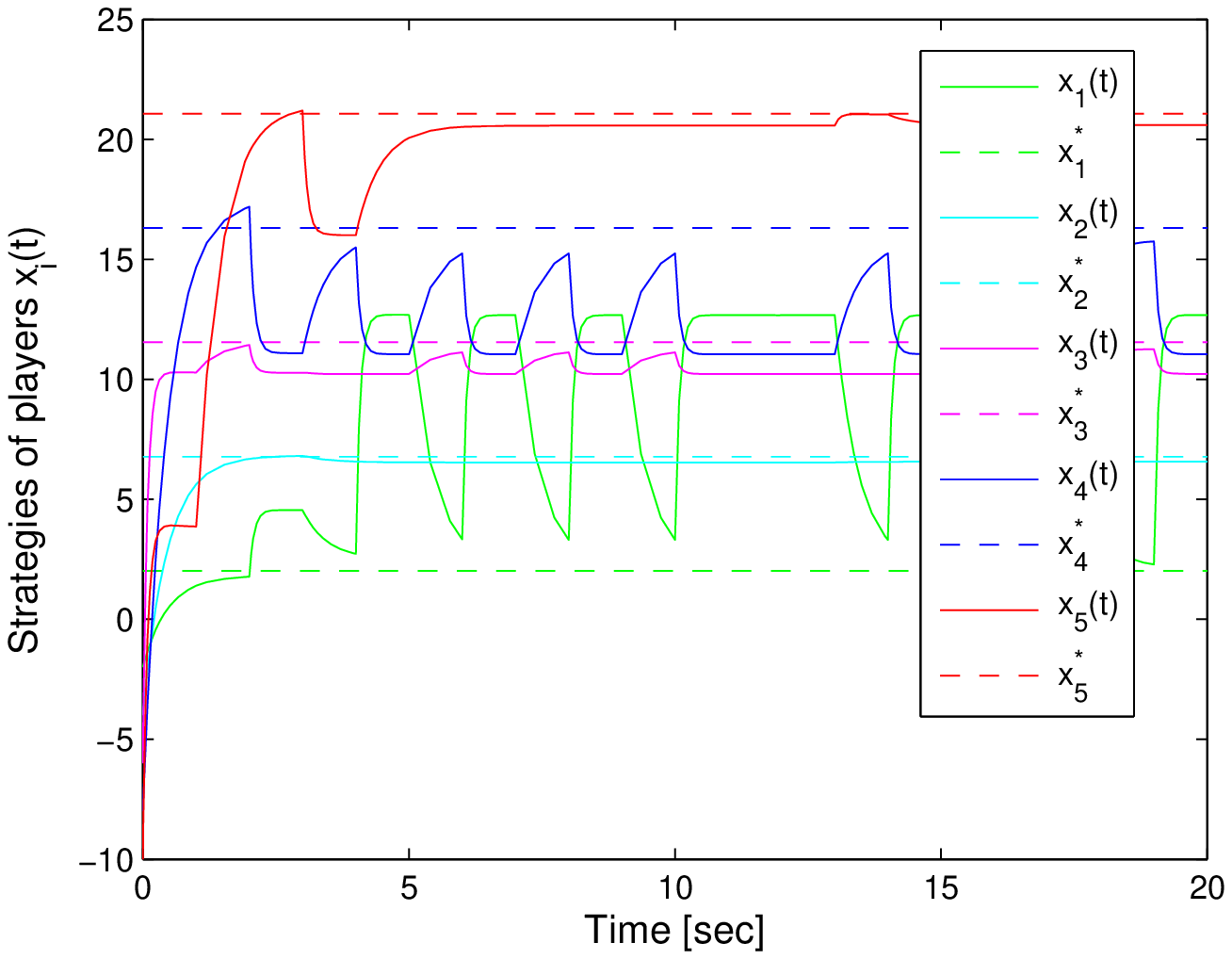} \label{A17}}
	\end{tabular}	
	\caption{The plot of players' strategies $x_{i}(t)$ produced by the proposed algorithm: (a) conditions (\ref{Condition1})-(\ref{Condition2}) hold; (b) conditions (\ref{Condition1})-(\ref{Condition2}) are not satisfied. }	
	\label{ComparisionAttacks} 
\end{figure}

\vspace{5pt}
\begin{example}
	(General Non-Quadratic Game)
\end{example}

In this example, we investigate a more general non-quadratic game, where the cost functions for each player $ i $ are given by 
\vspace{-3pt}
\begin{align}
J_{1}(x_{1},x_{-1})&= \frac{x^{2}_{1}}{2}+x_{1}\sum^{5}_{j=2}x_{j},   J_{2}(x_{2},x_{-2})= \frac{e^{\frac{x_{2}}{2}}}{2}+x_{2}x_{4}, \notag \\
J_{3}(x_{3},x_{-3})&=\frac{x^{2}_{3}}{2}+x_{1}^{3}, \
J_{4}(x_{4},x_{-4})= \ln(e^{x_{4}})+x^{2}_{4}+x^{3}_{3}, \notag \\
J_{5}(x_{5},x_{-5})&= x^{2}_{5}-5x_{5}+x^{3}_{1}x_{2} +x_{3}x^{4}_{4}, \ i=1,\cdots,5.
\end{align} 

Next, we perform the proposed algorithm for this non-quadratic game with the same simulation setting in the subsection V-A. By the calculation, the NE is $ x^{*} =\text{col} (-4.6589,4.1589,0,-2,2.5)$. The simulation result is shown in Fig. \ref{NonQuadraticGame}, and as we can see, under attacks, all players' estimates can reach consensus and converge to the NE exponentially.

\begin{figure}[!h]
	\centering
	\hspace{-0.5em}
	\includegraphics[width=9.0cm,height=9.0cm]{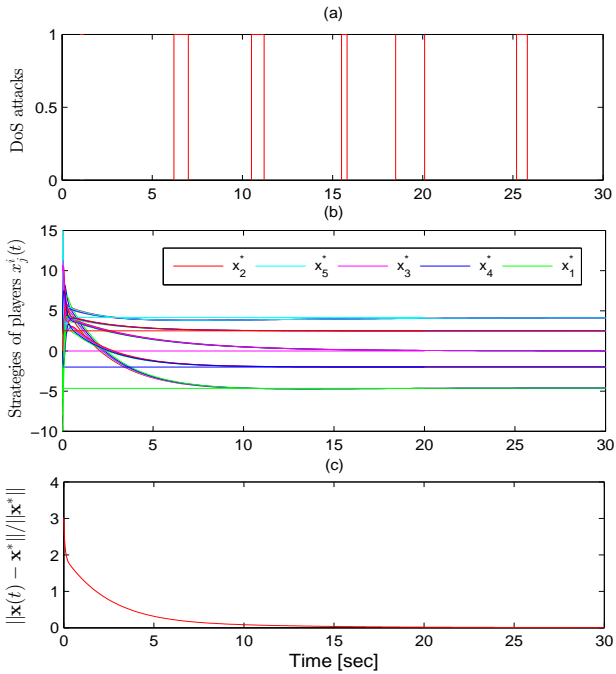}
	\caption{Simulated results of the proposed attack-resilient NE seeking algorithm in (\ref{NEController1})-(\ref{NE2}): (a) DoS attacks on various communication channels; (b) all players' estimated strategies $x^{i}_{j}(t), \ i, j\in \mathcal{V}$; and (c) relative errors of all players' actions.} 
	\label{NonQuadraticGame}
\end{figure}

\section{Conclusion} 
\vspace*{-2pt}
In this paper, an attack-resilient distributed algorithm has been presented for exponential NE seeking of non-cooperative games, where all players' strategies have been updated through a directed communication network subject to malicious DoS attacks. Under such an adversary network environment, the exponential convergence of the proposed distributed algorithm has been established through the explicit analysis of the attack frequency and duration. Moreover, in the absence of DoS attacks, the corollary has been provided, which can cover many existing results as special cases. The effectiveness of the developed approach has been illustrated by the numerical examples. Further work may consider distributed NE seeking problems for aggregative games with 
constraints.


\vspace*{-5pt}
\begin{thebibliography}{99}
\vspace*{-3pt}
\bibitem{Basar87AT} S. Li, T. Basar, ``Distributed algorithms for the
computation of noncooperative equilibria,'' \emph{Automatica}, 23, 523--533, 1987.
 
\bibitem{Shamma05TAC} J. S. Shamma, G. Arslan, ``Dynamic fictitious play, dynamic gradient play, and distributed convergence to Nash equilibria,'' \emph{IEEE Trans. Autom. Control}, 50(3): 312--327, 2005.

\bibitem{Johansson12TAC} M. S. Stankovic, K. H. Johansson, D. M. Stipanovic, ``Distributed seeking of Nash equilibria with applications to mobile sensor networks,'' \emph{IEEE Trans. Autom. Control}, 57(4): 904--919, 2012.

\bibitem{Basar12TAC} P. Frihauf, M. Krstic, T. Basar, ``Nash equilibrium seeking in noncooperative games,'' \emph{IEEE Trans. Autom. Control}, 57(5): 1192--1207, 2012.

\bibitem{Ye17Tcyber} 
S. Grammatico, ``Dynamic control of agents playing aggregative games with coupling constraints,'' \emph{IEEE Trans. Autom. Control}, 62(9):  4537--4548, 2017.  

\bibitem{Scutari04TIT} G. Scutari, F. Facchinei, J.   Pang, D. Pallomar, ``Real and complex monotone communication games,'' \emph{IEEE Trans. Inf. Theory}, 60(7): 400--409, 2014.


\bibitem{Marden07PhD} J. R. Marden, ``Learning in large-scale games and cooperative control,'' \emph{Ph.D. dissertation, University of California}, Los Angeles, CA, USA, 2007.


\bibitem{Ye17TAC} M. Ye, G. Hu, ``Distributed Nash equilibrium seeking by a consensus based approach,'' \emph{IEEE Trans. Autom. Control}, 62(9): 4811--4818, 2017. 

\bibitem{Ye19TAC} C. Sun, G. Hu, ``Distributed Nash equilibrium seeking for a generalized convex game with nonsmooth objective functions and certain nonsmooth constraints,'' \emph{International Conference on Information and Automation}, pp. 77--82, Fujian, China, Aug. 11--14, 2018.

\bibitem{Sun18} K. Lu, G. Jing, L. Wang, ``Distributed algorithms for searching generalized Nash equilibrium of noncooperative Games,'' \emph{IEEE Trans. Cybern.,} 49(6): 2362--2371, 2019. 

\bibitem{Pavel19TAC} D. Gadjov, L. Pavel, ``A passivity-based approach to Nash equilibrium seeking over networks,'' \emph{IEEE Trans. Autom. Control}, 64: 1077--1092, 2019. 

\bibitem{Liang17AT} S. Liang, P. Yi, Y. Hong, ``Distributed Nash equilibrium seeking for aggregative games with coupled constraints,'' \emph{Automatica}, 85: 179--185, 2017.


\bibitem{Deng19TNNLS}  Z. Deng, X. Nian, ``Distributed generalized Nash equilibrium seeking algorithm design for aggregative games over weight-balanced digraphs,'' \emph{IEEE Trans. Neur. Net. Lear.}, 30(3): 695--706, 2019.

\bibitem{Wang19ICCA} Y. Zhang, X. Wang, H. Ji, ``Distributed Nash equilibrium seeking in aggregative game with disturbance rejection,'' \emph{the 15th International Conference on Control and Automation}, Scotland, July 16-19, pp. 1091--1095, 2019.
 
\bibitem{Xie13Conf} X. Wang, N. Xiao, T. Wongpiromsarn, L. Xie, E. Frazzoli, D. Rus, ``Distributed consensus in noncooperative congestion games: an application to road pricing,'' \emph{in IEEE Int. Conf. Control Autom.}, pp. 1668–1673, 2013.


\bibitem{Pavel16TAC} F. Salehisadaghiani, L. Pavel, ``Distributed Nash equilibrium seeking: a gossip-based algorithm,'' \emph{Automatica}, 72(1): 209--216, 2016.



\bibitem{Pavel19AT} F. Salehisadaghiani, W. Shi, L. Pavel, ``Distributed Nash equilibrium seeking under partial-decision information via the alternating direction method of multipliers,'' \emph{Automatica}, 103: 27--35, 2019.

\bibitem{Pang20TAC} Y. Pang, G. Hu, ``Distributed Nash equilibrium seeking with limited cost function knowledge via a consensus-based gradient-free method,'' \emph{IEEE Trans.  Autom. Control}, DOI: 10.1109/TAC.2020.2995666, 2020.

\bibitem{Pang2000TAC} Y. Pang, G. Hu, ``Randomized gradient-free distributed optimization methods for a multi-agent system with unknown cost function,'' \emph{IEEE Trans. Autom. Control}, 65(1), pp. 333--340, 2020. 

%
%







\bibitem{Mo14SP} Y. Mo, J. P. Hespanha, B. Sinopoli, ``Resilient detection in the presence of integrity attacks,'' \emph{IEEE Trans. Signal Process}. 62(1): 31--43, 2014. 

\bibitem{Hu15IFAC} Z. Feng, G. Hu, ``Distributed coordinated control for multi-agent systems under two types of attacks with an application to power system,'' \emph{the 19$\mathit{th}$ IFAC World Congress,} August 24-29, South Africa, pages, 124–-130, 2014.

\bibitem{Hu15IJNRC} Z. Feng, G. Hu, G. Wen, ``Distributed consensus tracking for multi-agent systems under two types of attacks,'' \emph{Int. J. Robust. Nonlinear Control}, 26(5) : 896--918, 2015. 

\bibitem{Hu15Cyber} Z. Feng, G. Wen, G. Hu, ``Distributed secure control for multi-agent systems under strategic attacks,'' \emph{IEEE Trans. Cybern.,} 47(5): 1273--1284, 2017.

\bibitem{Hu17ACC} Z. Feng, G. Hu, ``Distributed secure average consensus for linear multi-agent systems under DoS attacks,'' \emph{American Control Conference,} May 24-26, Seattle, WA, USA, 2017. 

\bibitem{Hu19TCST} Z. Feng, G. Hu, ``Secure cooperative event-triggered control of linear multi-agent systems under DoS attacks,'' \emph{IEEE Trans. Contr. Syst. Tech.}, 28(3): 741--752 2020. 


\bibitem{Hu20Tcyber} W. Xu, G. Hu, D. W. C. Ho, Z. Feng, ``Distributed secure cooperative control under Denial-of-Service attacks from multiple adversaries,'' \emph{IEEE Trans. Cybern.,} 50(2): 3458--3467, 2020.


\bibitem{Feng17AT} S. Feng, P. Tesi, ``Resilient control under denial-of-service: Robust design,'' \emph{Automatica}, 79: 42--51, 2017.

\bibitem{Yang19TAC} A. Lu, G. Yang, ``Input-to-state stabilizing control for cyber-physical systems with multiple transmission channels under denial of service,'' \emph{IEEE Trans. Autom. Control}, 63(6): 1813--1820, 2018. 


\bibitem{Chen19SMC} B. Chen, D. Ho, W. Zhang, L. Yu, ``Distributed dimensionality reduction fusion estimation for cyber-physical systems under DoS attacks,'' \emph{IEEE Trans. Syst., Man, Cybern., Syst.}, 49(2): 455--468, 2019.

\bibitem{Zhang19Tcyber} D. Zhang, L. Liu, G. Feng, ``Consensus of heterogeneous linear multi-agent systems subject to aperiodic sampled-data and DoS attack,'' \emph{IEEE Trans. Cybern.}, 49(4):  1501--1511, 2019.


 

\bibitem{Su16ACC} L. Su, N. Vaidya, ``Multi-agent optimization in the presence of Byzantine adversaries: Fundamental limits,'' \emph{American Control Conference}, July 6-8, Boston, MA, USA, pages 7183–-7188,
2016. 


\bibitem{Sundaram19TAC} S. Sundaram, B. Gharesifard, ``Distributed optimization under adversarial nodes,'' \emph{IEEE Trans. Autom. Control,} 64(3): 1063--1076, 2019.

\bibitem{Zhao20TAC} C. Zhao, J. He, Q. Wang, ``Resilient distributed optimization algorithm against adversarial attacks,'' \emph{IEEE Trans. Autom. Control,} DOI: 10.1109/TAC. 2019.2954363, 2020. 
%

\bibitem{Hu19TCNS} Z. Feng, G. Hu, C. G. Cassandras, ``Finite-time distributed convex optimization for continuous-time multi-agent systems with disturbance rejection,'' \emph{IEEE Trans. Control Netw. Syst.,} 7(2): 686--698, 2020.


\bibitem{Başar95} T. Başar, G. J. Olsder. Dynamic noncooperative game theory. SIAM, 1995.

\bibitem{BookKhail} H. K. Khalil, Nonlinear systems. 3rd ed, Prentice-Hall, 2002.
 
\bibitem{BookLi14} Z. Li, Z. Duan. Cooperative control of multi-agent systems: a consensus region approach. CRC Press, 2014.

\bibitem{98Agiza} E. Ahmed, H. N. Agiza, ``Dynamics of a Cournot game with ncompetitors,'' \emph{Chaos Solitons \& Fractals}, 9(9): 1513--1517, 1998. 
	
\end{thebibliography}
\end{document}